%% file: main.tex
\PassOptionsToPackage{psdextra}{hyperref}

\documentclass[acmsmall]{acmart}

\input{macros}

\title{Efficient Grammar-Constrained Decoding via Parser Stack Classification}

\setcopyright{cc}
\setcctype{by}

\begin{document}

\author{Yongmin Li}
\orcid{0009-0001-3702-0043}
\affiliation{%
  \department{Key Laboratory of High Confidence Software Technology (PKU), MoE and School of Computer Science}
  \institution{Peking University}
  \city{Beijing}
  \country{China}
}
\email{liyongmin@pku.edu.cn}

\author{Yihong Dong}
\orcid{0000-0001-6228-4019}
\affiliation{%
  \department{School of Computer Science}
  \institution{Shanghai Jiao Tong University}
  \city{Shanghai}
  \country{China}
}
\email{dongyh@sjtu.edu.cn}

\author{Jia Li}
\orcid{0000-0002-5579-8852}
\affiliation{%
  \department{College of AI}
  \institution{Tsinghua University}
  \city{Beijing}
  \country{China}
}
\email{jia\_li@mail.tsinghua.edu.cn}

\author{Ge Li}
\correspondingauthor
\orcid{0000-0002-5828-0186}
\affiliation{%
  \department{Key Laboratory of High Confidence Software Technology (PKU), MoE and School of Computer Science}
  \institution{Peking University}
  \city{Beijing}
  \country{China}
}
\email{lige@pku.edu.cn}

\input{chapters/abstract}



\maketitle

\input{chapters/introduce}
\input{chapters/background}

\input{chapters/method}
\input{chapters/experiment}

\input{chapters/discuss}
\input{chapters/conclude}

\begin{acks}
This research is supported by the National Key R\&D Program under Grant No. 2023YFB4503801, the National Natural Science Foundation of China under Grant No. 62192733, 62192730, 62192731, the Beijing Major Science and Technology Project under Contract no. Z251100008425005.

This research is also supported by Tsinghua University - Keystone Electrical (Zhejiang) Co., Ltd. Joint Research Center for Embodied Multimodal Artificial Intelligence (JCEMAI) and the Beijing Natural Science Foundation under Grant No. 4264107.
\end{acks}

\input{chapters/reproduce}

\bibliographystyle{ACM-Reference-Format}
\bibliography{ref}

\end{document}

%% file: macros.tex
\usepackage{mathtools}
\usepackage{amsthm}
\usepackage{empheq}
\usepackage{booktabs}

\newtheorem{theorem}{Theorem}

\usepackage{algpseudocodex}
\usepackage{algorithm}

\usepackage{listings}

\definecolor{terminal}{RGB}{32,108,135}
\definecolor{string}{RGB}{143,2,18}
\definecolor{identifier}{RGB}{0,0,109}

\lstdefinelanguage{lark}{
    morekeywords={import,ignore,},
    otherkeywords={\%,->},
    morestring=[b][\color{string}]{"},
    moredelim=[is][\color{terminal}]{<}{>}
}

\usepackage{threeparttable}
\usepackage{longtable}
\usepackage{subcaption}
\usepackage{multirow}
\usepackage{makecell}

\usepackage{adjustbox}

\usepackage[inline]{enumitem}

\usepackage{wrapfig}

\usepackage{pifont}

\usepackage{tcolorbox}
\usepackage{colortbl}
\newcommand{\mycolor}{blue!25}
\newcommand{\chl}{\cellcolor{\mycolor}}

\newcommand\method{PSC}

\NewDocumentCommand\myarrow{s O{\xrightarrow} D<>{\mathcal} m m}{%
    #2[#3{#4}]{#5}%
    \IfBooleanTF{#1}%
        {\mathrel{\vphantom{\to}^*}}%
        {}%
}
\newcommand\mathsc[1]{\mathrm{\textsc{#1}}}
\newcommand\oncefsa[1]{I_{#1}}

\usepackage[normalem]{ulem}

\usepackage{tikz}
\usetikzlibrary{positioning}

%% file: chapters/abstract.tex
\begin{abstract}

LLMs are widely used to generate structured output like source code or JSON.
Grammar-constrained decoding (GCD) can guarantee the syntactic validity of the generated output, by masking out tokens that violate rules specified by a context-free grammar.
However, the online computational overhead of existing GCD methods, with latency typically scaling linearly with vocabulary size, limits the throughput of LLMs, especially for models with large vocabularies.
To address this issue, we propose \method{}, a novel grammar-constrained decoding method.
By combining acceptance conditions of all vocabulary tokens into a single classifier of the parser stack during preprocessing, \method{} can compute the complete vocabulary mask by checking the parser stack exactly once per decoding step, with time complexity independent of the vocabulary size.
Experiments show that \method{} computes masks up to 700× faster than baselines on complex programming language grammars, and up to 30× faster for schema-conformant JSON; end-to-end LLM throughput with \method{} approaches that of unconstrained decoding.
We analyze the preprocessing overhead for preprocessing providers and decoding users, and provide a break-even point analysis to help users decide whether to do preprocessing by themselves.

\end{abstract}

%% file: chapters/introduce.tex
\section{Introduction}
\label{sec:introduce}

In recent years, the ability for Large Language Models (LLMs) to generate structured output has been widely recognized and utilized~\citep{qwen2.5,llama3,gemma3}.
\begin{enumerate*}
\item Source code can be viewed as structured output that adheres to the syntax of programming languages, and LLM-based coding assistants, such as GitHub Copilot~\citep{copilot} and Cursor~\citep{cursor}, have been widely adopted by developers to assist in writing code to improve their productivity.
\item When LLMs are used as a tool, users often expect the generated output to conform to a specific format, such as Markdown or JSON with custom schemas~\citep{liu2024we,vllmteamStructuredOutputsVLLM,openaiStructuredModelOutputs}.
\end{enumerate*}
All of these applications rely on the ability of LLMs to generate output that adheres to a specific syntax.

However, generating in a structured format is complex, as it requires not only understanding the semantics of given input but also adhering to the specific grammars of target formats.
Since language models are essentially probabilistic models, there is no guarantee that the generated output will always conform to the required grammar.

To address this issue, \emph{grammar-constrained decoding}~(GCD)~\citep{gcd,picard,synchromesh,syncode} is proposed to ensure that the generated output always conforms to the specified context-free grammar.
A GCD method works by incorporating a grammar checker into the decoding process, as shown in Figure~\ref{fig:demo:gcd}.
At each decoding step, the checker determines which tokens in the vocabulary can be appended to the current prefix while not violating the grammar.
The logits generated by the language model are then masked to only allow the valid tokens, and the next token is generated by sampling from the masked logits.

\input{floats/demos}

The overhead of GCD is determined by the newly introduced step of validity calculation.
A naive implementation, as shown in Figure~\ref{fig:demo:others} would require calling the parser for \emph{every} token in the vocabulary to check its validity,
resulting in a time complexity proportional to $|\mathcal V|$ per decoding step, where $|\mathcal V|$ is the vocabulary size.
This can add significant overhead, especially for large vocabularies in modern language models,
e.g. 128k tokens in Llama-3~\citep{llama3}, 151k tokens in Qwen~series~\citep{qwen}, or 262k tokens in Gemma~3~\citep{gemma3}.
The overhead is particularly pronounced for smaller models, where the time taken by model inference is relatively small.

To speed up GCD, various techniques have been proposed in the literature,
summarized in Section~\ref{sec:related}.
However, they cannot fundamentally change the $\mathcal O(|\mathcal V|)$ worst-case time complexity.

We propose a novel GCD method, \textbf{P}arser \textbf{S}tack \textbf{C}lassification (\method{}), that replaces the repetitive runtime parsing over the whole vocabulary with a one-time classification of the current parser stack, as shown in Figure~\ref{fig:demo:ours}.
The checking process of the parser can be seen as a function of both the token and the state of the parser, which is usually a stack.
For each token, our method \method{} constructs a finite-state automaton (FSA) that represents the exact requirements on the parser stack to accept that token,
i.e., the FSA accepts a parser stack if and only if that token is accepted by a parser with that stack.
All these FSAs can then be combined into a single FSA that classifies the parser stack into a finite number of classes, each corresponding to a different vocabulary mask.
During decoding, we only need to check the parser stack exactly once per decoding step to get the vocabulary mask, which is ready to be applied to the logits.
This eliminates the need to call the parser for each token in the vocabulary, resulting in a significant speedup.

We conduct extensive experiments on grammar-constrained decoding in Java, Go, SQL, and schema-conformant JSON to evaluate the efficiency of our method.
Compared to the current state-of-the-art method, our method achieves up to 700 times speedup in mask computation on complex programming language grammars, 
and up to 30 times speedup for schema-conformant JSON generation.
In the end-to-end decoding throughput experiments, the throughput of \method{} approaches that of unconstrained decoding, and is significantly higher than the current state-of-the-art method, especially on smaller models and larger batch sizes.

We also analyze the preprocessing overhead of \method{}.
It is small for simple grammars, but can be significant for complex grammars.
Since the preprocessing results can be shared, the users are encouraged to use the preprocessing results provided by the model developers or other third parties when available, and only perform preprocessing by themselves when necessary.
Before doing preprocessing by themselves, the users should consider the break-even point analysis by estimating the expected decoding usage to decide whether to do preprocessing by themselves.

In summary, this paper makes the following contributions.
\begin{itemize}[wide]
  \item We propose a novel GCD method \method{} that leverages finite-state automata to classify parser states to use the precomputed vocabulary mask, significantly reducing the time overhead of grammar-constrained decoding.
  \item We provide a theoretical analysis of \method{}, justifying its correctness and providing a theoretical foundation for future research on grammar-constrained decoding.
  Specifically, we prove that the set of the parser stacks that can accept a given token can be formally described as a regular language.
  \item We demonstrate the efficiency of \method{} through extensive experiments on grammar-constrained decoding in Java, Go, Python, and schema-conformant JSON, achieving significant speedup in mask computation compared to existing techniques; end-to-end decoding throughput with \method{} approaches that of unconstrained decoding.
  \item We analyze the preprocessing overhead of \method{} for preprocessing providers and decoding users, and provide a break-even point analysis to help users decide whether to preprocess by themselves.
\end{itemize}

%% file: floats/demos.tex
\begin{figure}[t]
\centering

\begin{subfigure}{.8\textwidth}
\centering
\includegraphics*[width=\textwidth]{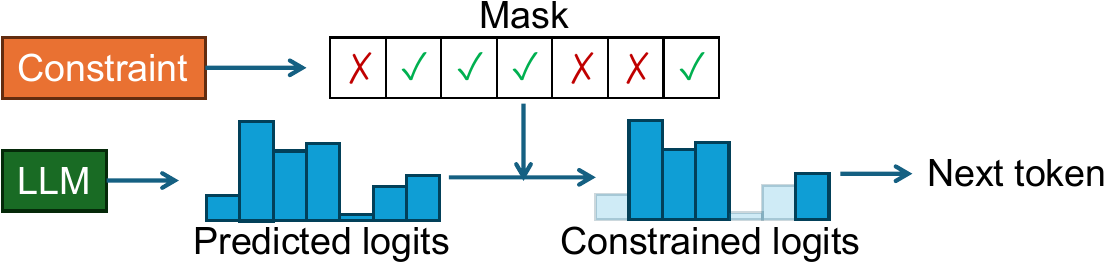}
\caption{Decoding step in grammar-constrained decoding (GCD).} \label{fig:demo:gcd}
\end{subfigure}\\
\subfloat[Naive GCD implementation.]{%
\includegraphics*[height=.1\textheight]{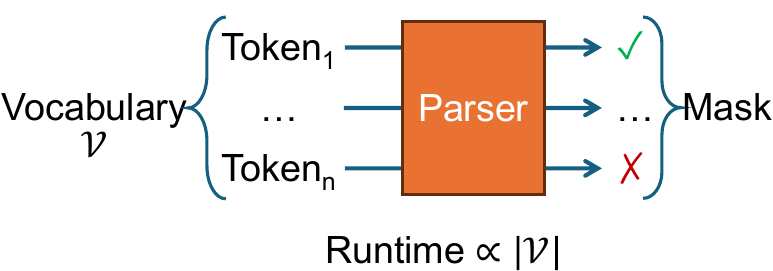}%
\label{fig:demo:others}%
}\hfill%
\subfloat[Our GCD method \method{}.]{
\includegraphics*[height=.1\textheight]{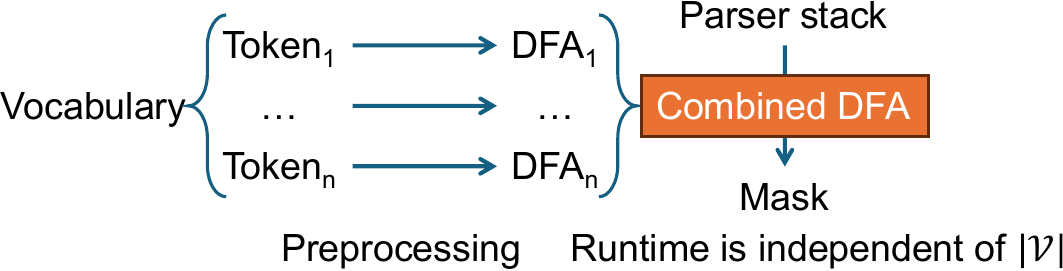}%
\label{fig:demo:ours}%
}
\caption{An illustration of grammar-constrained decoding, showing (a) the overall working process, (b) the naive implementation that directly simulates the PDA, and (c) our method \method{} that precomputes the DFA and the valid token masks.}
\label{fig:demo}
\end{figure}

%% file: chapters/background.tex
\section{Background and Related Work}\label{sec:background}

We introduce the task of grammar-constrained decoding in this section, give definitions of the concepts and symbols used in the paper, and then review the related work.
A symbol table is included in Table~\ref{tab:symbols} for most important symbols used in this paper.

\subsection{The Task: Grammar-Constrained Decoding}

Let $\Sigma$ be the character set used by the language model, e.g. the Unicode.
Given a prefix of tokens, the task of a language model is to generate the next-token distribution over the vocabulary $\mathcal V \subset \Sigma^+$.
For a language $L \subseteq \Sigma^*$, the task of \emph{constrained decoding} aims to generate a sample in $L$ from the language model.
In each step, given prefix $x \in \Sigma^*$, it calculates the set of valid tokens in $\mathcal{V}$, i.e. tokens that, when concatenated after $x$, become a prefix of some strings in $L$.
\begin{equation}
    c(x \in \Sigma^*) := \left\{v \in \mathcal{V} \middle| \exists y \in \Sigma^*, x v y \in L \right\}. \label{eq:cons}
\end{equation}

When the language $L$ is defined by a context-free grammar, the task is called \emph{grammar-constrained decoding}~\citep{syncode,koo2024automatabased,greatgramma,llguidance}.
Determining whether a string is syntactically valid usually involves two phases: lexical analysis and syntax analysis\footnote{To ease the presentation, the step of lexical analysis is omitted in previous sections. The term ``parser'' in previous sections should be realized as the combination of the lexer and the parser defined here, and the term ``parser stack'' should be realized as the concatenation of the lexer state and the parser state, which is a simple state without internal structure, and a stack, respectively.} \citep{ahoTheoryParsingTranslation1972}.
In lexical analysis, the lexer $\mathcal T$, usually modeled as a deterministic finite-state transducer (FST) \citep{ahoTheoryParsingTranslation1972,koo2024automatabased,greatgramma}, transduces the text $w \in \Sigma^*$ into a terminal sequence $\mathcal T(w) \in \Gamma^*$, where $\Gamma$ is the set of terminals.
In syntax analysis, the parser $\mathcal P$, usually modeled as a terminating deterministic push-down automaton (PDA) \citep{ahoTheoryParsingTranslation1972}, determines whether a terminal sequence $x \in \Gamma^*$ is valid, here written as $x \in \mathcal P$.
So we have
\begin{equation}
    w \in L \iff \mathcal T(w) \in \mathcal P.\label{eq:in-lang-mean}
\end{equation}

\subsection{Finite-State Transducer (FST)}
A \emph{finite-state transducer} (FST) \citep{ahoTheoryParsingTranslation1972} $\mathcal{T}$ is defined by a finite set of states $Q$, the input alphabet $\Sigma$, the output alphabet $\Gamma$, the start state $q_0 \in Q$, the final states $F \subseteq Q$, and transitions $\delta: Q \times \Sigma_\varepsilon \to 2^{\Gamma^* \times Q}$.
If $\delta(q, a) \ni (y, q')$,
we write $q \myarrow{T}{a:y} q'$.
We write $\myarrow*{}{}$ for consecutive transitions.
For $q \in Q$, we write $q \myarrow*{T}{\varepsilon:\varepsilon} q$.
For $q \myarrow*{T}{s:x} q'$ and $q' \myarrow{T}{t:y} q''$, we write $q \myarrow*{T}{st:xy} q''$.

$\mathcal{T}$ is \emph{deterministic} if, for all $q \in Q$, 
either $(|\delta(q, a)| \leq 1, \forall a \in \Sigma)$ and $\delta(q, \varepsilon) = \varnothing$, or $(\delta(q, a) = \varnothing, \forall a \in \Sigma)$ and $\delta(q, \varepsilon) = 1$.
Informally, this means that from any state, for any given string, there is exactly one possible outcome.

For $w \in \Sigma^*, q \in Q$, we define $\mathcal T_q(w)$ as $\left\{v \in \Gamma^* \middle| \exists q' \in F, q \myarrow*{T}{w:v} q' \right\}$,
meaning the possible outcomes when we feed $w$ into $\mathcal T$ starting from the state $q$.
When $\mathcal T$ is deterministic and $v \in \mathcal T(w)$, we also write $\mathcal T(w) = v$.
For $W \subseteq \Sigma^*$, we define $\mathcal T_q(W)$ as $\bigcup_{w \in W} \mathcal T_q(w)$.
$q$ defaults to $q_0$ when omitted.

We call a state $q \in Q$ \emph{stable} if the FST does not need to take any immediate action on $q$, i.e. $\delta(q, \varepsilon) = \varnothing$.
If $q \myarrow*{T}{s:t} q'$ and $q'$ is stable, we also write $q \myarrow*[\xRightarrow]{T}{s:t} q'$.

Given two FSTs $\mathcal S$ and $\mathcal T$ where the output alphabet of $\mathcal S$ is the input alphabet $\mathcal T$, $\Gamma^{\mathcal S} = \Sigma^{\mathcal T}$,
their \emph{composition} is a new FST $\mathcal S \circ \mathcal T$ by feeding the output of $\mathcal S$ into the input of $\mathcal T$.

\subsection{Finite-State Automaton (FSA)}

A \emph{finite-state automaton} (FSA) $\mathcal A$ can be defined by removing all output labels from an FST.
We say $\mathcal A$ \emph{accept}s $w \in \Sigma^*$ from state $q$ if $\mathcal A_q(w) \neq \varnothing$, and write $w \in \mathcal A_q$, where $q$ defaults to $q_0$ when omitted.
Two FSAs are \emph{equivalent} if they accept the same language.

$\mathcal A$ is \emph{deterministic} if there is no $\varepsilon$ transition in $\delta$.
Every nondeterministic FSA can be \emph{determinized} into an equivalent deterministic FSA \citep{hopcroftIntroductionAutomataTheory1979},
and every deterministic FSA can be \emph{minimized} into an equivalent deterministic FSA with the smallest number of states \citep{hopcroftIntroductionAutomataTheory1979}.

The \emph{union} of two FSAs $\mathcal A$ and $\mathcal B$ is a new FSA $\mathcal{A \cup B}$ that accepts any sequence that is accepted by either $\mathcal A$ or $\mathcal B$.
The \emph{concatenation} of two FSAs $\mathcal A$ and $\mathcal B$ is a new FSA $\mathcal{AB}$ that accepts any sequence that
can be split into two parts $x = y z$,
where $\mathcal A$ accepts the first part $y$ and $\mathcal B$ accepts the second part $z$.

\subsection{Push-Down Automaton (PDA)}
\label{bg:syntax}

A \emph{push-down automaton} (PDA) \citep{ahoTheoryParsingTranslation1972,hopcroftIntroductionAutomataTheory1979,caucalTransitionGraphsAutomata1991} $\mathcal P$ is defined by
the input alphabet $\Gamma$,
the stack alphabet $\Pi$,
the initial stack $\gamma_0 \in \Pi^2$,
the final states $F \subseteq \Pi$, 
and a finite set of transitions $\delta: \Pi^2 \times \Gamma_\varepsilon \to 2^{\Pi^+}$.
\footnote{Note that the definition of transitions here merges the states and the stack symbols in the traditional definition of PDA, but they are equivalent if we treat the stack top symbol as the state.}
If $\delta(\alpha, a) \ni \beta$,
we write $\alpha \myarrow{P}{a} \beta$.
If $\alpha \myarrow{P}{a} \beta$, for any $\gamma \in \Pi^*$, we also write $\alpha \gamma \myarrow{P}{a} \beta \gamma$.
We write $\myarrow*{}{}$ for consecutive transitions.
For $\gamma \in \Pi^+$, we write $\gamma \myarrow*{P}{\varepsilon} \gamma$.
If $\alpha \myarrow{P}{a} \beta$, $\beta \myarrow*{P}{w} \gamma$, we write $\alpha \myarrow*{P}{aw} \gamma$.

Informally, a PDA is \emph{deterministic}, if from any stack, for any given string, there is exactly one possible outcome.
$\mathcal P$ is \emph{deterministic} if, for any $\alpha \in \Pi^2$, either $(\left|\delta(\alpha, a)\right| \leq 1, \forall a \in \Gamma)$ and $\delta(\alpha, \varepsilon) = \varnothing$, or $(\delta(\alpha, a) = \varnothing, \forall a \in \Gamma)$ and $\left|\delta(\alpha, \varepsilon)\right| = 1$.

A deterministic PDA is \emph{terminating}, if for any stack, it does not make an endless sequence of $\varepsilon$-input transitions.\footnote{The definition is slightly different in the cited references; nevertheless, their proof works on this definition.}
Every deterministic PDA can be transformed into another equivalent deterministic terminating PDA \citep{sipserIntroductionTheoryComputation2013,hopcroftIntroductionAutomataTheory1979}.

We call a state $\beta = \beta_1 \dots \beta_n \in \Pi^+$ \emph{stable}
if $\beta_1$ is a final state, i.e. $\beta_1 \in F$, or the PDA is waiting to read one more symbol, i.e. $\exists a \in \Gamma, \delta(\beta_1 \beta_2, a) \neq \varnothing$.
If $\alpha \myarrow*{P}{w} \beta$ and $\beta$ is stable,
we also write $\alpha \myarrow*[\xRightarrow]{P}{w} \beta$.
In practice, parser in a stable state is ready to consume the next input symbol, or has reached an accepting stack.

We also write $w \in \mathcal P$ for $w \in \Gamma^*$ if $\exists X \in F, \gamma \in \Pi^*$, such that $\gamma_0 \myarrow*[\xRightarrow]{P}{w} X\gamma$.

\input{floats/symbols}

\subsection{Related Work}\label{sec:related}

There are several types of existing techniques to speed up grammar-constrained decoding: vocabulary preprocessing, lexer preprocessing, and parser preprocessing.

\textbf{Vocabulary preprocessing}~\citep{synchromesh,domino,llguidance} exploits the fact that the vocabulary is built by BPE~\citep{bpe0,bpe1}, and for each token, its prefix is also a token in the vocabulary.
If the prefix token is rejected, then the longer token must also be rejected.
So we can check the vocabulary hierarchically, and only check the tokens whose prefixes are not rejected.

\textbf{Lexer preprocessing}~\citep{domino,greatgramma,llguidance,syncode} maps each token to a terminal sequence during preprocessing, and then the parser is only called on the terminal sequences.
This reduces the number of parser calls, as different tokens may share the same terminal sequence.
The mask can be precomputed for each terminal sequence, combined at runtime to get the valid token mask.
Syncode~\citep{syncode} further approximates the terminal sequences by only considering the first 2 terminals of each token, removing the need for dynamic parsing using the lookaheads of the LR(1) parser at the cost of allowing certain invalid tokens to be accepted.

\textbf{Parser preprocessing}~\citep{xgrammar,greatgramma} classifies the vocabulary into three sets for each parser state: context-independent accepted, context-independent rejected, and context-dependent.
This allows us to reduce the number of parser calls by only checking the context-dependent tokens.

\method{} is a novel parser preprocessing technique that eliminates all context-dependent tokens.
As seen from the example in Figure~\ref{fig:demo:run}, the validity of tokens often depends on the deep parser stack structure, so existing techniques cannot pre-classify them into context-independent accepted or rejected sets, and thus require multiple parser calls per decoding step.
In constrast, we discover and prove that their acceptance conditions on the stack structure can be modeled with FSAs, so calling parsers are not needed at runtime.
By further pre-combining the FSAs, we only need to check the parser state once per decoding step, which is the theoretical minimum for any GCD method.

\input{floats/example-runtime}

%% file: floats/symbols.tex
\begin{table}[t]
\centering
\caption{Symbols and their meaning.\label{tab:symbols}}
\begin{tabular}{ll}
\toprule
Symbol & Meaning \\
\midrule
$\varepsilon$ & empty string \\
$ab$ & concatenation of strings $a$ and $b$ \\
$AB$ & concatenation of languages $A$ and $B$ \\
$A_\varepsilon$ & $A \cup \{\varepsilon\}$ \\
$A^*$ & Kleene star of language $A$ \\
$A^+$ & $A^* \setminus \{\varepsilon\}$ \\
\midrule
$\Sigma$ & the character set (usually the Unicode) used by the language model \\
$\Gamma$ & the terminal set of the grammar \\
$\mathcal{V}$ & \makecell[l]{the vocabulary of the language model, a finite subset of $\Sigma^+$\\\hspace{1em} such that every string over $\Sigma$ can be tokenized as a string over $\mathcal{V}$} \\
$\mathcal{T}$ & the lexing FST, transduces string over $\Sigma$ to terminal sequence over $\Gamma$ \\
$Q$ & the finite set of states of the lexing FST $\mathcal T$ \\
$\mathcal{P}$ & the parsing PDA, accepts terminal sequences in the language \\
$\Pi$ & the stack alphabet of the PDA \\
\midrule
$p \myarrow*[\xRightarrow]{T}{v:x} q$ & \makecell[l]{successful transition of the lexing FST $\mathcal T$ from state $p \in Q$ to \emph{stable} state $q \in Q$\\\hspace{1em} when reading token $v \in \mathcal V$ and outputting terminal sequence $x \in \Gamma^*$} \\
$\alpha \myarrow*[\xRightarrow]{P}{w} \beta$ & \makecell[l]{successful transition of the parsing PDA $\mathcal P$ from stack $\alpha \in \Pi^*$ to \emph{stable} stack $\beta \in \Pi^*$\\\hspace{1em} when reading terminal sequence $w \in \Gamma^*$} \\
\bottomrule
\end{tabular}
\end{table}

%% file: floats/example-runtime.tex
\begin{figure}[t]
\begin{subfigure}{\textwidth}
\centering
\begin{minipage}{.47\textwidth}
\begin{lstlisting}[language=java, linewidth=\textwidth]
public class C {
  public static int gcd(int x, int y) {
    if (y == 0) return x!\big|!
\end{lstlisting}
\end{minipage}\hfill%
\begin{tabular}{c|c|c|c|c|c|c|c|c|}
    \multicolumn{9}{l}{$\xleftarrow{\text{top of the stack}}$}\\
    \hline
    2 & 79 & 540 & 483 & 326 & 406 & 574 & 11 & $\dots$\\
    \hline
  \end{tabular}
\caption{Example prefix and corresponding parser stack.} \label{fig:demo:prefix}
\end{subfigure}\\
  \begin{subfigure}{.47\textwidth}
    \includegraphics*[width=\textwidth]{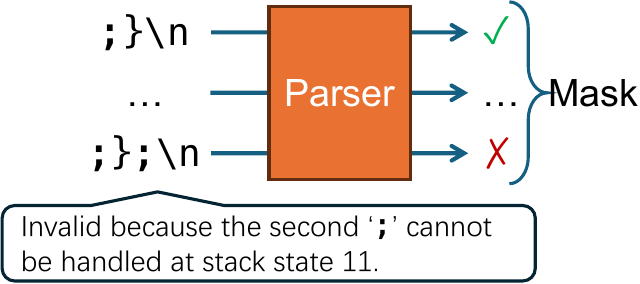}
    \caption{How existing GCD methods handle the example prefix.
    Because the validity of context-dependent tokens depends on the deep structure of the parser stack, they need to run the parser at runtime.}\label{fig:demo:stack}
  \end{subfigure}
\hfill%
\begin{subfigure}{.51\textwidth}
\usetikzlibrary {automata}
\begin{tikzpicture}[node distance=0.6cm, state/.style={circle, draw, minimum size=0.6cm}, font=\small]
  \node[state, initial] (0) {};
  \node[state] (2) [right=of 0] {};
  \node[state] (79) [right=of 2] {};
  \node[state] (540) [right=of 79] {};
  \node[state] (483) [below=of 540] {};
  \node[state] (326) [left=of 483] {};
  \node[state] (574) [above=of 79] {};
  \node[state, draw=none] (dots) [right=of 574] {$\dots$};
  \node[state, draw=none] (mask) [right=of dots] {mask};
  \path[->] (0) edge node [above] {2} (2)
    (2) edge node [above] {79} (79)
    (79) edge node [above] {540} (540)
    (540) edge node [right] {483} (483)
    (483) edge node [above] {326} (326)
    (326) edge node [left] {406} (79)
    (79) edge node [left] {574} (574)
    (574) edge node [above] {11} (dots)
    (dots) edge (mask);
\end{tikzpicture}
\caption{\method{} analyses the exact conditions of the PDA to accept each token in preprocessing, so it can directly get the valid token mask by classifying the parser stack at runtime without repeated simulation.
  Only parts relevent to the example are shown.} \label{fig:demo:run:ours}
\end{subfigure}
\caption{An example of how \method{} works during runtime decoding, compared with a naive GCD implementation.
\method{}'s improvement comes from exact analysis of the parser stack.}\label{fig:demo:run}
\end{figure}
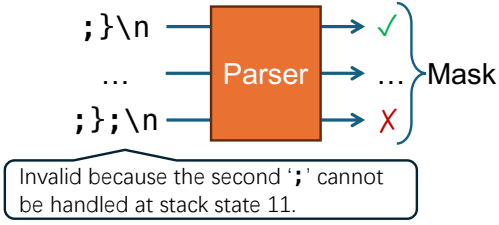
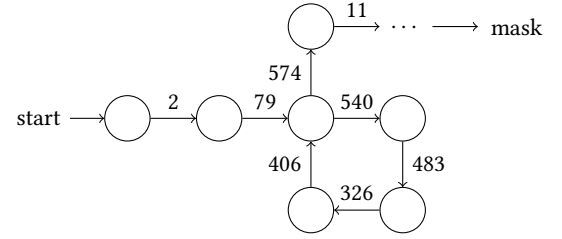

%% file: chapters/method.tex
\section{\method{}: Parser Stack Classification} \label{sec:method}

The preprocessing of the lexer $\mathcal T$ is described in Section~\ref{sec:lex_pre}, and the other parts of this section are dedicated to the preprocessing of the parser $\mathcal P$.

\subsection{Lexical Preprocessing}\label{sec:lex_pre}

Lexical preprocessing is not our focus in this paper,
so we reuse the lexical preprocessing in GreatGramma~\citep{greatgramma},
and conclude it here as a prelude to \method{}.

For token $v \in \mathcal V$, lexer state $q \in Q$, if lexing the token $v$ from state $q$ using the lexer $\mathcal T$ generates the terminal sequence $x \in \Gamma^*$, and the lexer transits to a stable state $p$, i.e. $q \myarrow*[\xRightarrow]{T}{v:x} p$, we define the set of \emph{realizable terminal sequences} $R_q(v)$ as $x T_p$, representing all possible terminal prefixes that can be generated from a string starting with $v$,
where $T_p$ is the finite set of all possible terminal prefixes from state $p$: $\mathcal T_p(\Sigma^*) = T_p \mathcal T(\Sigma^*)$.
$R_q(v)$ is computed for all $q \in Q, v \in \mathcal V$ during preprocessing.

\subsection{Overview of Syntactic Preprocessing}\label{sec:overview}

Given a valid prefix $x \in \Sigma^*$, we run the lexer $\mathcal T$ from its initial state $q_0$ to produce a terminal sequence $z \in \Gamma^*$ and another stable lexer state $q$, i.e. $q_0 \myarrow*[\xRightarrow]{T}{x:z} q$.
We then run the parser $\mathcal P$ from its initial stack $\gamma_0$ to receive the terminal sequence $z$, generating another stable stack $\alpha$, i.e. $\gamma_0 \myarrow*[\xRightarrow]{P}{z} \alpha$.

We now introduce the simplification of the condition in the GCD definition in Equations~\ref{eq:cons} and~\ref{eq:in-lang-mean} \textbf{from previous work~\citep{greatgramma}}.
For any token $v \in \mathcal V$, to determine whether $v$ is valid, we can rewrite the condition in terms of realizable terminal sequences,
\begin{equation}
    \exists y \in \Sigma^*, \mathcal T(xvy) \in \mathcal P
    \iff \exists w \in R_q(v), \exists \beta \in \Pi^+, \alpha \myarrow*[\xRightarrow]{P}{w} \beta, \label{eq:simp}
\end{equation}
which means if the parser can reach a stable stack $\beta$ from the current stack $\alpha$ by reading in any realizable terminal sequence $w$ from $R_q(v)$.
The simplification is based on the common assumption that,
if the parser reads in a certain terminal sequence and enters a stable stack,
then we do not need to worry about the rest of the input, and
there always exists a terminal sequence produced by the lexer that can ensure the whole text is accepted.

For $w \in \Gamma^*$, we define $P_w(\alpha)$ for the calculation in the last step of Equation~\ref{eq:simp},
\begin{equation}
    P_{w \in \Gamma^*}(\alpha \in \Pi^+) := \left\{ \beta \in \Pi^+ \middle| \alpha \myarrow*[\xRightarrow]{P}{w} \beta \right\}. \label{eq:ind}
\end{equation}
\textbf{How to efficiently calculate $P_w(\alpha)$ is the key difference between \method{} and previous methods.}
In existing work, the calculation of $P_w$ is almost always dynamic:
one has to calculate $P_w(\alpha)$ for the current $\alpha$ and every possible $w \in R_q(\mathcal V)$.
While existing methods in Section~\ref{sec:related} employ precomputation to optimize certain cases, they still fundamentally require worst-case $\mathcal O(|R_q(\mathcal V)|)$ time for dynamic parsing if correctness is not sacrificed.

In this work, \method{} proposes a totally different approach, modeling $P_w$ as a deterministic finite-state transducer (FST),
which reads the stack sequence $\alpha \in \Pi^+$, and then outputs the sequence $\beta \in \Pi^+$ if the parser $\mathcal P$ can reach $\beta$ from $\alpha$ by reading $w$, or rejects the input stack otherwise.
In other words, the behavior of the parser $\mathcal P$ on any stack $\alpha$ by reading $w$ is fully encoded in the FST $P_w$.

This gives us several benefits.
\begin{enumerate*}
\item Because each $P_w$ reads and outputs a sequence of stack symbols, they can be composed to create larger FSTs: $P_t \circ P_s = P_{st}$.
\item Because the realizable terminal sequences are known during precomputation, the exact validity condition of each vocabulary is therefore known, their combinations can be precomputed, and we only need to go through the current stack once during runtime.
\item All possible masks can also be precomputed, eliminating the mask generation overhead during decoding.
\end{enumerate*}

The challenge here is whether and how each $P_w$ can be constructed as a deterministic FST.
This is not straightforward because of the presence of $\varepsilon$-transitions in the PDA $\mathcal P$.
To address this, we first construct $P_\varepsilon$ to handle all $\varepsilon$-transitions, and then construct $P_w$ for any $w \in \Gamma^*$ based on $P_\varepsilon$.

\subsection{FST of $\varepsilon$ Transitions}

In this section, we construct the FST $P_\varepsilon$.
Its input should be a stack, and \textbf{the output is its stabilized version, by repeatedly executing all needed $\varepsilon$ transitions on the stack}.
The start state is $\varepsilon$, and the final state is a special state \textsc{Final}.
The set of all states is the minimum closure of the transitions defined below, where each state represents the current known stack top.
\begin{subequations}\begin{align}
\forall X \in \Pi, &&\forall \alpha \in \Pi^*, &&\alpha &\myarrow<>{P_\varepsilon}{X:\varepsilon} \alpha X, &&\text{if } \left|\alpha \right| < 2 \text{ and } \alpha_{[0]} \notin F_{\mathcal P}; \label{eq:fst:epsilon:readup}\\
&&\forall \alpha \in \Pi^*, && \alpha &\myarrow<>{P_\varepsilon}{\varepsilon:\alpha} \mathsc{Final}, &&\text{if } \alpha_{[:2]} \myarrow{P}{a} \beta, \exists a \in \Gamma \text{ or } \alpha_{[0]} \in F_{\mathcal P}; \label{eq:fst:epsilon:done}\\
&&\forall \alpha \in \Pi^*, && \alpha &\myarrow<>{P_\varepsilon}{\varepsilon:\varepsilon} \beta\alpha_{[2:]}, &&\text{if } \alpha_{[:2]} \myarrow{P}{\varepsilon} \beta; \label{eq:fst:epsilon:execute}\\
\forall X \in \Pi, &&&&\mathsc{Final} &\myarrow<>{P_\varepsilon}{X:X} \mathsc{Final}. \label{eq:fst:epsilon:final}
\end{align}\label{eq:fst:epsilon}\end{subequations}
There are four types of transitions in $P_\varepsilon$.
\begin{enumerate*}
\item[\eqref{eq:fst:epsilon:readup}] If one cannot determine whether the stack is stable from the stack top $\alpha$, it transits to a new state by reading the next stack symbol.
\item[\eqref{eq:fst:epsilon:done}] If the stack top $\alpha$ is stable, it transits to the \textsc{Final} state to output the final stable stack.
\item[\eqref{eq:fst:epsilon:execute}] Otherwise, it simulates the transition of $\mathcal P$ on the current stack top $\alpha$, and transits to a new state representing the new stack top after executing the $\varepsilon$ transition.
\item[\eqref{eq:fst:epsilon:final}] In the \textsc{Final} state, it always outputs the input stack unchanged.
\end{enumerate*}
We have the following theorem regarding the correctness of the construction above.

\begin{theorem}\label{thm:epsilon}
Equations~\ref{eq:fst:epsilon} constructs a finite-state transducer $P_\varepsilon$ as defined in Equation~\ref{eq:ind}.
\end{theorem}

\begin{proof}
First we show that the states of $P_\varepsilon$ are finite, i.e., starting from the initial state $\varepsilon$, only a finite number of new states are needed to construct a closure of all transitions.
Consider the two equations that introduce new states: Equation~\ref{eq:fst:epsilon:readup} and Equation~\ref{eq:fst:epsilon:execute}.
\begin{enumerate*}
\item[\eqref{eq:fst:epsilon:readup}] The length of its right-hand side is at most $2$, so only a finite number of states are needed here.
\item[\eqref{eq:fst:epsilon:execute}] This is the only place that can introduce new states without length limit.
However, because the parser $\mathcal P$ is terminating, by definition in Section~\ref{bg:syntax}, for any stack configuration, there will not be an endless sequence of $\varepsilon$ transitions,
so this process can only introduce a finite number of new states before reaching a stable stack and transiting to \textsc{Final}.
\end{enumerate*}

The correctness of Equations~\ref{eq:fst:epsilon} can be naturally deduced by its construction, because it just simulates the behavior of the parser $\mathcal P$ with the current stack top, and only outputs stable stacks.
\end{proof}

$P_\varepsilon$ is an important building block in the construction of other $P_w, w \in \Gamma^+$.
For any stack $\alpha$, $P_\varepsilon(\alpha)$ gives the stabilized version of $\alpha$,
so the FST composed after $P_\varepsilon$ does not need to handle $\varepsilon$ transitions,
and we can compose $P_\varepsilon$ after other FSTs to meet the stability requirement in the definition of $P_w$.

\subsection{FST for Any Terminal Sequence}

After constructing $P_\varepsilon$, the construction of $P_w$ for any terminal sequence $w \in \Gamma^+$ is direct.
We first construct an FST $\tilde P_a$ for every $a \in \Gamma$ that \textbf{simulates a single transition labeled $a$}, i.e. outputting $\beta$ for the input stack $\alpha$ if $\alpha \myarrow{P}{a} \beta$, i.e.,
\begin{equation}
\tilde P_a(\alpha \in \Pi^+) := \left\{ \beta \in \Pi^+ \middle| \alpha \myarrow{P}{a} \beta \right\}.\label{eq:single}
\end{equation}
Note that the output stack is not required to be stable, so we add a tilde above $P$ to indicate this difference from $P_w$.
The start state is $\varepsilon$,
the final state is $\textsc{Final}$,
and transitions are defined as follows.
\begin{align}
  \varepsilon &\myarrow<>{\tilde P_a}{X:\varepsilon} X \myarrow<>{\tilde P_a}{Y:\varepsilon} XY \myarrow<>{\tilde P_a}{\varepsilon:\beta} \textsc{Final}, \forall XY \myarrow{P}{a} \beta; &\textsc{Final} &\myarrow<>{\tilde P_a}{X:X} \mathsc{Final}, \forall X \in \Pi. \label{eq:fst:single}
\end{align}
For any terminal sequence $w = w_1 \dots w_n \in \Gamma^+$, the FST $P_w$ can be constructed as follows. 
\begin{equation}
  P_w = P_\varepsilon \circ \tilde P_{w_1} \circ P_\varepsilon \circ \dots \circ P_\varepsilon \circ \tilde P_{w_n} \circ P_\varepsilon.\label{eq:compose}
\end{equation}
Intuitively, the input stack $\alpha$, is first passed to $P_\varepsilon$ to get a stable stack,
and then passed to $\tilde P_{w_1}$ to get the stack after reading $w_1$,
and then passed to $P_\varepsilon$ to get the stabilized version, etc,
until it is passed to $\tilde P_{w_n}$ and stabilized with $P_\varepsilon$.

We have the following theorem regarding the correctness of the construction above.

\begin{theorem}\label{thm:terminal_sequence}
The above construction of $P_w$ meets the definition in Equation~\ref{eq:ind}.
\end{theorem}

\begin{proof}
The construction of $\tilde P_a$ in Equations~\ref{eq:fst:single} simulates one $a$-labelled transition of the parser $\mathcal P$ on the current stack top, so it meets Equation~\ref{eq:single}.

The process of the parser processing $w = w_1 \dots w_n$ can be decomposed into a sequence of unconditional $\varepsilon$-transitions, followed by $w_i$-labelled transitions, followed by another sequence of unconditional $\varepsilon$-transitions. Each $w_i$-labelled transition is simulated by the corresponding $\tilde P_{w_i}$, and the unconditional $\varepsilon$-transitions are handled by $P_\varepsilon$.
Therefore, the composition in Equation~\ref{eq:compose} correctly simulates the parser $\mathcal P$ processing the terminal sequence $w$ on the input stack $\alpha$, and produces the stabilized output stack $\beta$ if it exists.
\end{proof}

When calculating the mask, we only care about whether $P_w(\alpha) \neq \varnothing$.
Removing all the output labels from $P_w$ gives us a finite-state automaton, hereafter named $A_w$.

\subsection{One-Pass FSA for Mask Selection}

After constructing $P_w$ for all realizable terminal sequences $w \in R(\mathcal V)$,
we can now consider simplifying the constraint calculation over the whole vocabulary $\mathcal V$.
Recall Equation~\ref{eq:cons}, combined with Equation~\ref{eq:simp} and $A_w$, we have the following equation,
\begin{equation*}
    c(x \in \Sigma^*) = \left\{ v \in \mathcal V \middle| \exists w \in R_q(v), P_w(\alpha) \neq \varnothing\right\} = \left\{ v \in \mathcal V \middle| \exists w \in R_q(v), \alpha \in A_w\right\},
\end{equation*}
where $q$ and $\alpha$ as defined in Section~\ref{sec:overview} are only dependent on $x$.

In $c(x)$, we want to know whether any of the $A_w$ accepts $\alpha$, where $w \in R_q(v)$.
This can be achieved by constructing the union of different $A_w$, i.e., $\bigcup_{w \in R_q(v)} A_w$.

But to get the whole mask $c(x)$, we need to check for every $v \in \mathcal V$, whether $\alpha$ is accepted by $\bigcup_{w \in R_q(v)} A_w$, which is inefficient.
To address this issue, we can integrate the checking of $v$ and $q$ into the FSA.
Introduce the notation $\oncefsa{a}$ for an FSA that accepts only $a$ once.
For every $q \in Q$ and $v \in \mathcal V$, we can concatenate $\oncefsa{q}$ before $\bigcup_{w \in R_q(v)} A_w$ to check whether the current state is $q$, and then concatenate $\oncefsa{v}$ after $\bigcup_{w \in R_q(v)} A_w$ to check whether the candidate token is $v$.

By unioning the results for all $v \in \mathcal V$ and $q \in Q$, we construct an FSA $\mathcal A$ that accepts the sequence $q \alpha v$ only if $v$ is a valid token for the lexer state $q$ and stack $\alpha$,
\begin{align}
    \mathcal A &:= \bigcup_{v \in \mathcal V} \bigcup_{q \in Q} \bigcup_{w \in R_q(v)} \oncefsa{q} A_w \oncefsa{v}, &
    c(x) &= \left\{ v \in \mathcal V \middle| q \alpha v \in \mathcal A \right\}, \label{eq:final_c}
\end{align}
where $\mathcal A$ should be determinized and minimized.
This gives us the following theorem.

\begin{theorem}\label{thm:regular}
All (lexer state, parser stack) pairs that accept a given token form a regular language.
\end{theorem}
\begin{proof}
Because $\mathcal A$ is constructed as a FSA, the language recognized by $\mathcal A$ is regular.
The language of all valid (lexer state, parser stack) pairs for a given token $v \in \mathcal V$ can be obtained by reversing the language recognized by $\mathcal A$, taking the Brzozowski derivative~\citep{brzozowskiderivative} with respect to $v$, and then reversing it back.
Since the class of regular languages is closed under these operations~\citep{hopcroftIntroductionAutomataTheory1979}, the resulting language is also regular.
\end{proof}

In Equation~\ref{eq:final_c}, we can precompute all possible result of $c$, i.e. all possible vocabulary masks, by considering acceptable vocabulary set $\mathcal A_s := \left\{v \in \mathcal V \middle| s \myarrow{A}{v} f^\mathcal A \right\}$ for every state $s$ in $\mathcal A$ where $f^\mathcal A$ is the final state of $\mathcal A$.

We summarize the offline construction process of \method{} in Algorithm~\ref{alg:offline}, and the online execution process in Algorithm~\ref{alg:online}.
In Algorithm~\ref{alg:online}, both the lexing step~\ref{step:lex} and the parsing step~\ref{step:parse} are standard in grammar-constrained decoding, and can be incrementally maintained.
In Step~\ref{step:fsa}, the FSA $\mathcal A$ is run on the stack $\alpha$ and the lexer state $q$ to get the state $s$,
only requiring $\mathcal O(\left|\alpha\right|)$ time.
Step~\ref{step:mask} can be precomputed to be $\mathcal O(1)$ at runtime.

\begin{figure}[ht]
\begin{minipage}{0.5\textwidth}
\input{floats/offline}
\end{minipage}\hfill
\begin{minipage}{0.48\textwidth}
\vfill
\input{floats/summary}
\vfill
\end{minipage}
\end{figure}

%% file: floats/offline.tex
\begin{algorithm}[H]
\caption{Offline constructon in \method{}}\label{alg:offline}
\begin{algorithmic}[1]

\Function{OfflineConstruction}{$\mathcal T, \mathcal P, \mathcal V$}

\State $P_\varepsilon \gets \Call{EpsilonFST}{\mathcal P}$
\ForAll{$a \in \Gamma$}
  \State $\tilde P_a \gets \Call{TerminalFST}{\mathcal P, a}$
\EndFor
\ForAll{$w = w_1 \dots w_n \in R(\mathcal V)$}
  \State $P_w \gets P_\varepsilon \circ \tilde P_{w_n} \circ P_\varepsilon \circ \cdots \circ P_\varepsilon \circ \tilde P_{w_1} \circ P_\varepsilon$
  \State $A_w \gets \Call{RemoveOutput}{P_w}$
\EndFor
\State $\mathcal A \gets \bigcup_{v \in \mathcal V} \bigcup_{q \in Q} \bigcup_{w \in R_q(v)} \oncefsa{q} A_w \oncefsa{v}$
\State $\mathcal A \gets \Call{Minimize}{\mathcal A}$
\State \Return $\mathcal A$

\EndFunction

\end{algorithmic}
\end{algorithm}

%% file: floats/summary.tex
\begin{algorithm}[H]
\caption{Online execution of \method{}}\label{alg:online}
\begin{algorithmic}[1]

\Function{OnlineExecution}{$\mathcal T, \mathcal P, \mathcal A, x$}

\State $q_0^\mathcal T \myarrow*[\xRightarrow]{T}{x:z} q$ \label{step:lex}
\State $\gamma_0 \myarrow*[\xRightarrow]{P}{z} \alpha$ \label{step:parse}
\State $q_0^\mathcal A \myarrow*[\xRightarrow]{A}{q \alpha} s$ \label{step:fsa}
\State \Return $\mathcal A_s$ \label{step:mask}

\EndFunction

\end{algorithmic}
\end{algorithm}

%% file: chapters/experiment.tex
\section{Experiments}\label{sec:exp}

In this section, we conduct experiments to answer the following research questions:
\begin{enumerate}[label=(RQ\arabic*),ref=RQ\arabic*]
  \item Does \method{} produce \textbf{correct} masks for grammar-constrained decoding? \label{rq:correct} 
  \item How much \textbf{online overhead} does \method{} introduce? \label{rq:mask_overhead}
  \item How does \method{} affect the \textbf{end-to-end decoding throughput}? \label{rq:throughput}
  \item Does \method{} improve the performance on \textbf{downstream tasks}? \label{rq:downstream}
\end{enumerate}

We design the experiments focusing on two main aspects:
\begin{itemize}[wide]
  \item (\ref{rq:correct}, \ref{rq:downstream}) The \textbf{usefulness} of \method{}.
    \ref{rq:correct} checks its correctness, and \ref{rq:downstream} validates its usefulness in downstream tasks by replicating results from previous work.
    Note that all grammar-constrained decoding methods in principle compute the \textbf{same} valid token masks, so they should all produce correct masks if implemented correctly, and perform equally well on downstream tasks.
  \item (\ref{rq:mask_overhead}, \ref{rq:throughput}) The \textbf{efficiency} of \method{}.
    \ref{rq:mask_overhead} focuses on the online overhead of \method{}, which is the main contribution of this paper, and \ref{rq:throughput} measures the end-to-end decoding throughput to show the overall impact of mask computation efficiency on decoding speed.
\end{itemize}

\input{chapters/experiments/setup}

\input{chapters/experiments/correctness}

\input{chapters/experiments/overhead}

\input{chapters/experiments/throughput}

\input{chapters/experiments/downstream}

%% file: chapters/experiments/setup.tex
\subsection{Experimental Setup}\label{sec:exp:setup}

In this section, we describe the experimental setup used to answer the research questions and evaluate \method{} against several state-of-the-art grammar-constrained decoding (GCD) methods.

\paragraph{Methodology}
In \ref{rq:correct}, \ref{rq:mask_overhead}, and \ref{rq:throughput}, we prepare the oracle sequences and use \textbf{teacher-forcing} during evaluation,
i.e., we always use the oracle next token at each decoding step,
\textbf{to ensure that all methods are evaluated under the same conditions.}
This is important because the context and the length of the generated sequences can significantly affect the mask computation process, making it hard to fairly compare different methods otherwise.

In \ref{rq:downstream}, we evaluate the downstream task performance of \method{} by letting the model generate tokens autoregressively without teacher-forcing.

\paragraph{Datasets}
There is no standard benchmark for evaluating grammar-constrained decoding methods.
We choose two representative tasks that require grammar-constrained decoding: code generation in Java, Go, and SQL, and JSON generation with specified JSON schemas.
For each task, we construct the evaluation dataset, with 1000 samples for each programming language and 1000 schemas for schema-conformant JSON generation.

For the Java, Go, and SQL datasets, we obtain their Lark grammars from the previous work Syncode~\citep{syncode}.
Because the grammar formats for XGrammar and Formatron are different from Lark, we manually convert the Lark grammars to respective formats for each baseline.
For each programming language, we take the first 1000 samples from the Stack dataset~\citep{thestack} that can be successfully parsed by the Lark parser to construct the evaluation dataset.

For the JSON Schemas, we use the benchmark dataset MaskBench~\citep{maskbench}, an extension of JSON Schema Bench~\citep{jsonschemabench} by adding schema conformant and non-conformant JSON instances to each schema.
We generate the Lark grammar from the JSON schemas using the script provided in MaskBench,
and only use the schemas in MaskBench where the Lark parser can successfully parse all the conformant JSON instances and reject all the non-conformant JSON instances.
We then randomly sample 1000 schemas for evaluation.
There are a total of 1337 positive samples and 2072 negative samples in the JSON dataset, and each schema has at least one positive sample and one negative sample.

\paragraph{Baselines}
We consider several recent state-of-the-art grammar-constrained decoding methods with open-source implementations as baselines, including XGrammar~\citep{xgrammar}, GreatGramma~\citep{greatgramma}, Formatron~\citep{formatron}, and LLGuidance~\citep{llguidance}.
\begin{itemize}[wide]
  \item XGrammar~\citep{xgrammar} uses a character-level non-deterministic PDA\footnote{In the latest implementation that we use in the experiments, this has been changed to an Earley~\citep{earley} parser.}. For each state, it
  precomputes the context-independent accepted and rejected tokens, and only calls the parser for the context-dependent tokens.
  It also includes many optimizations to speed up the parsing process, such as expanding rule contexts to further reduce context-dependent tokens, and persistent execution stacks to reduce memory consumption.
  \item GreatGramma~\citep{greatgramma} uses a lexer and a parser. It
  converts each token into all possible terminals sequences and reduces the number of parser calls by sharing the parser calls among tokens with the same terminal sequence.
  After computing the accepted terminal sequences, it maps them back to the original tokens.
  It also precomputes the context-dependent and context-independent terminal sequences for each parser state.
  \item Formatron~\citep{formatron} uses an Earley parser. It dynamically identifies and eliminates invalid or redundant parser states during parsing, and also uses context-independent prefiltering to reduce the number of tokens that need to be checked by the parser.
  \item LLGuidance~\citep{llguidance} uses a lexer and an Earley parser. It
  organizes the vocabulary into a trie, and skips the whole subtree if the prefix token is rejected.
  It also leverages the lexer on the vocabulary to pre-identify the terminal sequences.
\end{itemize}

There are other baselines not included in our experiments.
Domino~\citep{domino} and \citeauthor{koo2024automatabased}'s work~\citep{koo2024automatabased} are excluded due to the lack of an open-source implementation.
Pre$^3$~\citep{pre3} is excluded because it relies on unreleased code.
Outlines~\citep{outlines} and llama.cpp~\citep{llamacpp} are excluded because they are too slow to practically evaluate on such a large dataset, which is aligned with the findings in previous work~\citep{greatgramma,formatron}.
Syncode~\citep{syncode} is excluded because it rejects the oracle token in too many cases, probably because of its buggy implementation of unlexed characters handling.

\paragraph{Models}
We evaluate three recent open-source LLM series with different vocabulary sizes: Llama~3.2~\citep{llama3} (128k tokens), Qwen2.5~\citep{qwen,qwen2.5}~ (151k tokens), Gemma~3~\citep{gemma3} (262k tokens).
In RQ1 and RQ2, only the tokenizer is needed to calculate the mask given a prefix, so the model size does not affect the results.
In RQ3 and RQ4, we state the model used in each experiment, which is always the smallest model in each series to highlight the overhead of constraint decoding, and we also include the results on Qwen2.5~7B in RQ3 to see the effect of model size on throughput.

\paragraph{Implementation}
We implement \method{} in roughly 1100 lines of Python.
Similar to previous work~\citep{syncode,greatgramma}, we use the Lark parser\citep{lark} to construct the lexer and the LALR(1) parser from the grammar due to the one-to-one correspondence between the LR parsers and the DPDA~\citep{knuthTranslationLanguagesLeft1965}.
As described in Section~\ref{sec:lex_pre}, we reuse the lexer construction in GreatGramma~\citep{greatgramma}, which this is not our focus here.

\paragraph{Execution environment}
We conduct our experiments on a machine with 8 NVIDIA A100 GPUs (40 GB Memory), 2 Intel Xeon Gold 6348 CPUs (2.6GHz, 56 cores), and 512 GB RAM.
To ensure fairness, we run all the experiments with a single GPU and a single CPU thread.
For all the baselines, we use the latest versions of the respective libraries during the experiments.
The environment configuration file for the libraries used in the experiments is included in the replication package.

%% file: chapters/experiments/correctness.tex
\subsection{\ref{rq:correct}: Correctness of Mask Computation}\label{sec:exp:correctness}

\input{floats/passrate}

\paragraph{Metrics}
For each method, we measure the \textbf{sample pass rate}, i.e. the proportion of samples that are correctly processed by each method. Positive samples (including all programming language samples and some JSON samples) are counted as passed if they are not rejected, and negative samples (including some JSON samples) are counted as passed if they are rejected.

\paragraph{Results}
The sample pass rates are shown in Table~\ref{tab:passrate}.
\method{} achieves a high pass rate across all grammars and models, very close to 100\%.
We analyze the error cases and find that they are all directly rejected by the GreatGramma lexer we adopt, but \textbf{no error is caused by \method{} itself}.

There are several reasons for the lexer rejections, mainly falling into two categories:
\begin{enumerate*}
  \item The lexer assumes that each Unicode character is contained within a single token, but some tokenizers may split a Unicode character across multiple tokens, causing the lexer to misinterpret the input.
    This issue can be fixed by improving the lexer implementation to properly handle Unicode characters.
  \item The lexer is designed with a one-character lookahead, which may lead to misinterpretation of certain character sequences in some edge cases.
    This issue can be avoided by improving the grammar to better account for such sequences, but can only be fixed by fundamentally redesigning the lexer, which is out of the scope of this paper. But this kind of error is rare in practice.
\end{enumerate*}

Most of the baselines also achieve high pass rates (mostly above 99\%) on all the grammars and models, showing that existing GCD methods can generally produce correct masks.
One of the exceptions is Formatron on SQL, which has a significantly lower pass rate (around 67\%).
After further investigation, we find that Formatron refuses to accept the column alias in the query, resulting in frequent rejections.
The exact reason is unclear, but it is probably due to their buggy implementation of handling certain grammar constructs.

\begin{tcolorbox}[size=title]
  \textbf{Answer to \ref{rq:correct}}: \method{} produces correct masks for grammar-constrained decoding, achieving nearly 100\% sample pass rate across all grammars and models tested.
\end{tcolorbox}

%% file: floats/passrate.tex
\begin{table}[bt]
\centering
\caption{How many samples are correctly processed by each method. The symbol {\color{red} X} indicates the parser reports an error during mask calculation. Text in \textbf{bold} indicates the best performance, and text in \underline{underline} indicates the second best performance.}
\label{tab:passrate}
\begin{threeparttable}[b]
\newcommand{\oldtabcolsep}{\tabcolsep}
\setlength{\tabcolsep}{7pt}
\begin{tabular}{cl|rrrr}
\toprule
\multirow{2}{*}{Model} & \multirow{2}{*}{Method} & \multicolumn{4}{c}{Grammar} \\
& & Java & Go & SQL & JSON Schemas \\\midrule
\multirowcell{5}{Llama 3} & XGrammar & \underline{99.7\%} & \textbf{100.0\%} & \underline{99.7\%} & \textbf{100.0\%} \\
& Formatron & 99.4\% & {\color{red} X}\tnote{a} & 67.0\% & {\color{red} X}\tnote{a} \\
& GreatGramma & \textbf{100.0\%} & \underline{99.9\%} & 97.1\% & 99.6\% \\
& LLGuidance & \textbf{100.0\%} & {\color{red} X}\tnote{b} & \textbf{99.9\%} & \underline{99.9\%} \\
& \chl \method{} (Ours) & \chl \textbf{100.0\%} & \chl \underline{99.9\%} & \chl \textbf{99.9\%} & \chl 99.6\% \\\midrule
\multirowcell{5}{Qwen2.5} & XGrammar & \underline{99.7\%} & \textbf{100.0\%} & \underline{99.7\%} & \textbf{100.0\%} \\
& Formatron & 99.4\% & {\color{red} X}\tnote{a} & 67.0\% & {\color{red} X}\tnote{a} \\
& GreatGramma & \textbf{100.0\%} & \underline{99.9\%} & 97.2\% & 99.6\% \\
& LLGuidance & \textbf{100.0\%} & {\color{red} X}\tnote{b} & \textbf{99.9\%} & \underline{99.9\%} \\
& \chl \method{} (Ours) & \chl \textbf{100.0\%} & \chl \underline{99.9\%} & \chl \textbf{99.9\%} & \chl 99.6\% \\\midrule
\multirowcell{5}{Gemma 3} & XGrammar & \underline{99.7\%} & \textbf{100.0\%} & \underline{99.7\%} & \textbf{100.0\%} \\
& Formatron & \textbf{100.0\%} & {\color{red} X}\tnote{a} & 67.3\% & {\color{red} X}\tnote{a} \\
& GreatGramma & \textbf{100.0\%} & \underline{99.9\%} & 97.2\% & 99.6\% \\
& LLGuidance & 94.0\% & {\color{red} X}\tnote{b} & \textbf{99.9\%} & \underline{99.9\%} \\
& \chl \method{} (Ours) & \chl \textbf{100.0\%} & \chl \underline{99.9\%} & \chl \textbf{99.9\%} & \chl 99.6\% \\\bottomrule
\end{tabular}
\setlength{\tabcolsep}{\oldtabcolsep}
\begin{tablenotes}
\item[a] Formatron stops responding on multiple samples, so we terminate the process. 
\item[b] LLGuidance reports \texttt{ParserTooComplex} error.
\end{tablenotes}
\end{threeparttable}
\end{table}

%% file: chapters/experiments/overhead.tex
\subsection{\ref{rq:mask_overhead}: Overhead of Mask Computation}\label{sec:exp:overhead}

\paragraph{Metrics}
For each method, we measure the \textbf{average overhead}, i.e. the time of computing the CPU mask tensor at each decoding step.
We ignore the time taken to transfer the mask to the GPU and apply it to the logits, because this is the same for all methods\footnote{In \method{}, one can preload all the mask tensors on the GPU memory before decoding begins, eliminating the transfer overhead.
However, the transfer overhead is usually too small to justify the extra GPU memory usage.}.

\input{floats/overhead}

\paragraph{Results}
The average overhead is shown in Table~\ref{tab:overhead}.
\method{} significantly outperforms all the baselines across all grammars and models,
being 300 to 700 times faster on complex programming language grammars compared to the best baseline, LLGuidance,
and generally 30 times faster on the relatively simple JSON schema grammar.

The overhead of baselines is significantly larger on programming languages than that on JSON schemas.
This is because they need to repeat the parsing process for many times per decoding step to compute the valid token masks, so their overhead increases when the parsing process becomes more expensive due to increased grammar complexity.
In contrast, the overhead of \method{} remains stable across different grammars, because it only needs to parse the current token and traverse the parser stack once per decoding step, making this overhead mostly negligible.

The overhead of baselines is significantly larger on Gemma 3 than that on the other models\footnote{The performance of LLGuidance on JSON Schemas is roughly the same across different models, probably because the grammar is simple enough that the vocabulary size does not significantly affect the performance.}, because Gemma 3 has a much larger vocabulary than the other two models, and the time complexity of the baselines is roughly linear to the vocabulary size.
In contrast, the performance of \method{} is stable across different models, because its time complexity is independent of the vocabulary size.

\begin{tcolorbox}[size=title]
  \textbf{Answer to \ref{rq:mask_overhead}}: \method{} significantly reduces the overhead of mask computation in GCD, achieving a speedup of up to 700x on complex programming language grammars and 30x on simpler JSON schemas compared to the best baseline.
\end{tcolorbox}

%% file: floats/overhead.tex
\begin{table}[th]
\centering
\caption{Average overhead (microseconds) of computing the mask per token in GCD tasks using different methods. The symbol {\color{red} X} indicates the parser reports an error during mask calculation. Text in \textbf{bold} indicates the best performance, and text in \underline{underline} indicates the second best performance.}
\label{tab:overhead}
\begin{tabular}{cl|rrrr}
\toprule
\multirow{2}{*}{Model} & \multirow{2}{*}{Method} & \multicolumn{4}{c}{Grammar} \\
& & Java & Go & SQL & JSON Schemas \\\midrule
\multirowcell{5}{Llama 3\\$|\mathcal V| = 128256$} & XGrammar & 309514.2 & 281500.9 & 324663.6 & 26257.6 \\
& Formatron & 393540.4 & {\color{red} X} & 303974.1 & {\color{red} X} \\
& GreatGramma & 21402.8 & \underline{27220.9} & 20954.5 & 8556.2 \\
& LLGuidance & \underline{1352.4} & {\color{red} X} & \underline{826.1} & \underline{72.7} \\
& \chl \method{} (Ours) & \chl \textbf{2.4} & \chl \textbf{2.5} & \chl \textbf{2.6} & \chl \textbf{2.3} \\\midrule
\multirowcell{5}{Qwen2.5\\$|\mathcal V| = 151665$} & XGrammar & 302421.9 & 278333.4 & 299968.6 & 29470.0 \\
& Formatron & 378921.0 & {\color{red} X} & 253311.9 & {\color{red} X} \\
& GreatGramma & 24570.9 & \underline{27649.8} & 24053.9 & 11038.4 \\
& LLGuidance & \underline{1408.2} & {\color{red} X} & \underline{865.2} & \underline{72.1} \\
& \chl \method{} (Ours) & \chl \textbf{2.4} & \chl \textbf{2.5} & \chl \textbf{2.5} & \chl \textbf{2.2}\\\midrule
\multirowcell{5}{Gemma 3\\$|\mathcal V| = 262145$} & XGrammar & 649321.1 & 458952.0 & 416625.0 & 54164.4 \\
& Formatron & 696144.6 & {\color{red} X} & 444810.0 & {\color{red} X} \\
& GreatGramma & 43218.6 & \underline{48354.0} & 40026.9 & 25717.3 \\
& LLGuidance & \underline{1802.6} & {\color{red} X} & \underline{1180.8} & \underline{72.1} \\
& \chl \method{} (Ours) & \chl \textbf{2.3} & \chl \textbf{2.5} & \chl \textbf{2.4} & \chl \textbf{2.3} \\\bottomrule
\end{tabular}
\end{table}

%% file: chapters/experiments/throughput.tex
\input{floats/throughput}

\subsection{\ref{rq:throughput}: End-to-End Throughput}\label{sec:exp:throughput}

\paragraph{Settings}
In this experiment, we run the actual model inference using the vLLM library~\citep{vllm}, and measure the throughput, i.e., the number of tokens processed per second, of the entire decoding process, considering both the time taken by mask computation and model inference.
This metric reflects the overall efficiency of each method in practical usage.
We only consider the accepted samples of each method.
We compare \method{} with the fastest baseline LLGuidance in the previous section, and also include the throughput when not using any constraint decoding method as a reference, under various batch sizes of 1, 2, 4, ..., 256.

We use the smallest models in the three model series to highlight the overhead of constraint decoding.
On these models, the model inference time is relatively small, making the overhead of constraint decoding more pronounced.
Specifically, we use Llama~3~1B, Qwen~2.5~0.5B, and Gemma~3~270M.
We also include Qwen~2.5~7B to see the effect of model size on throughput.

\paragraph{Results}
The end-to-end throughput results on all datasets are present in Figures~\ref{fig:throughput}.

On all datasets, \method{} consistently outperforms LLGuidance across all models and batch sizes, and is very close to the performance of unconstrained decoding.
The difference is more pronounced on the programming languages, where more complex grammar leads to higher overhead of LLGuidance.

As the model size increases, the difference in throughput becomes smaller, because the model inference time becomes more dominant.
However, since smaller models are less capable, grammar-constrained decoding is probably more useful for smaller models to ensure the syntactic correctness.

As the batch size increases, the difference in throughput becomes larger, because the average model inference time per token decreases with larger batch sizes,
indicating that the overhead introduced by constraint decoding becomes more pronounced at larger batch sizes.

\begin{tcolorbox}[size=title]
  \textbf{Answer to \ref{rq:throughput}}: \method{} can deliver significantly higher end-to-end throughput than existing GCD methods across different models and batch sizes, and its throughput approaches that of unconstrained decoding, especially on smaller models and larger batch sizes.
\end{tcolorbox}

%% file: floats/throughput.tex
\newcommand{\oldtabcolsep}{\tabcolsep}
\setlength{\tabcolsep}{.5\tabcolsep}

\begin{figure}[bh]
\centering

\begin{tabular}{ccccc}
& Gemma 3 270M & Qwen2.5 0.5B & Llama 3.2 1B & Qwen2.5 7B \\

Java &
\raisebox{-.5\totalheight}{\includegraphics[width=.21\textwidth]{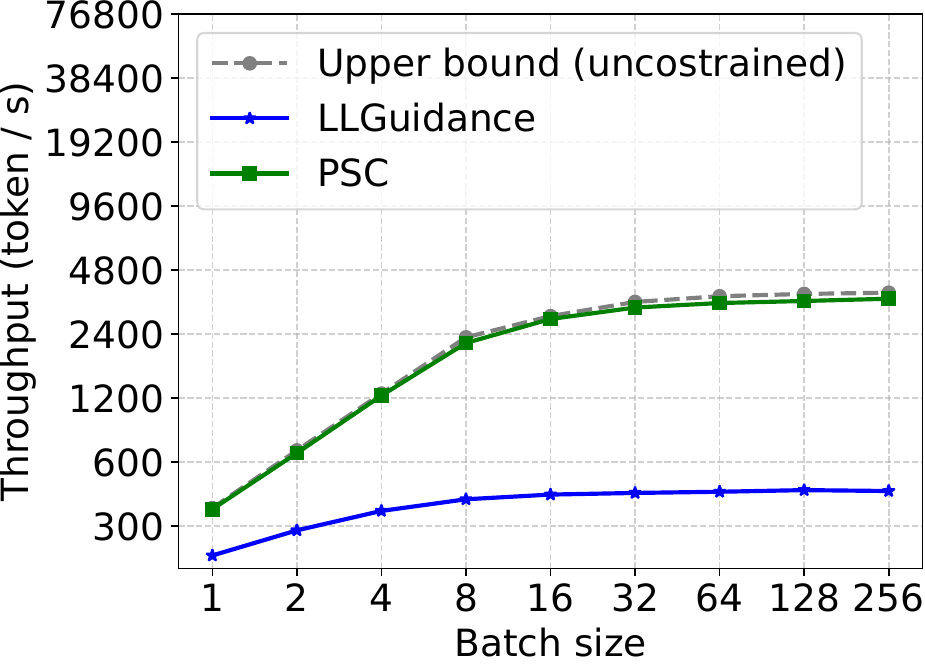}}
&
\raisebox{-.5\totalheight}{\includegraphics[width=.21\textwidth]{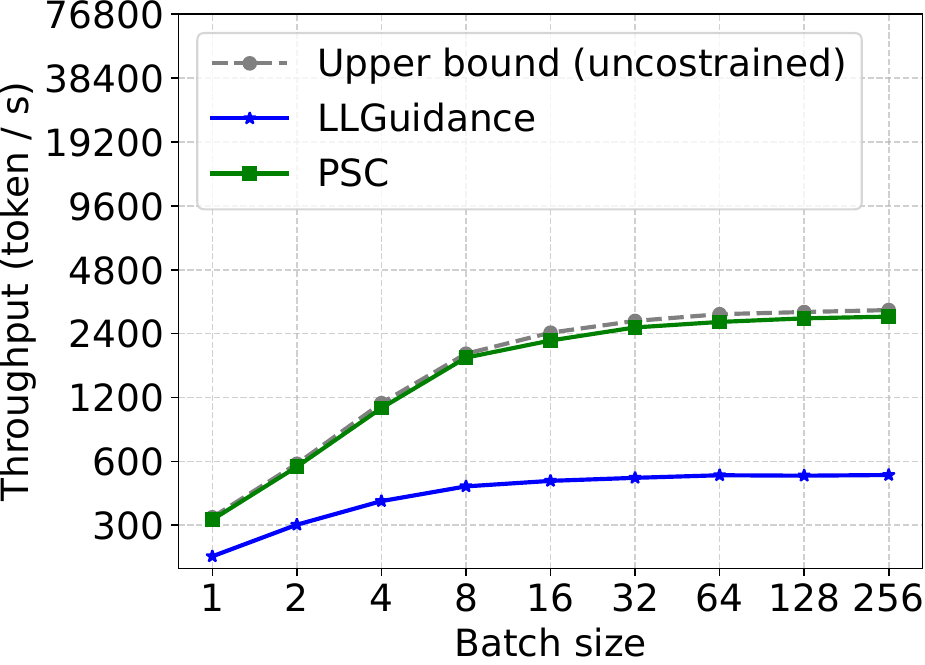}}
&
\raisebox{-.5\totalheight}{\includegraphics[width=.21\textwidth]{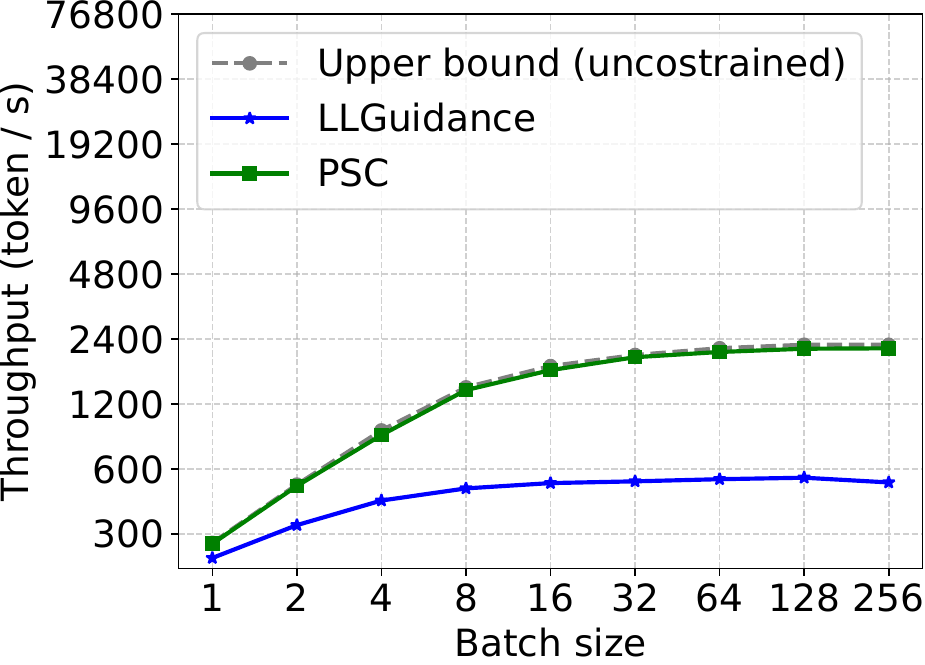}}
&
\raisebox{-.5\totalheight}{\includegraphics[width=.21\textwidth]{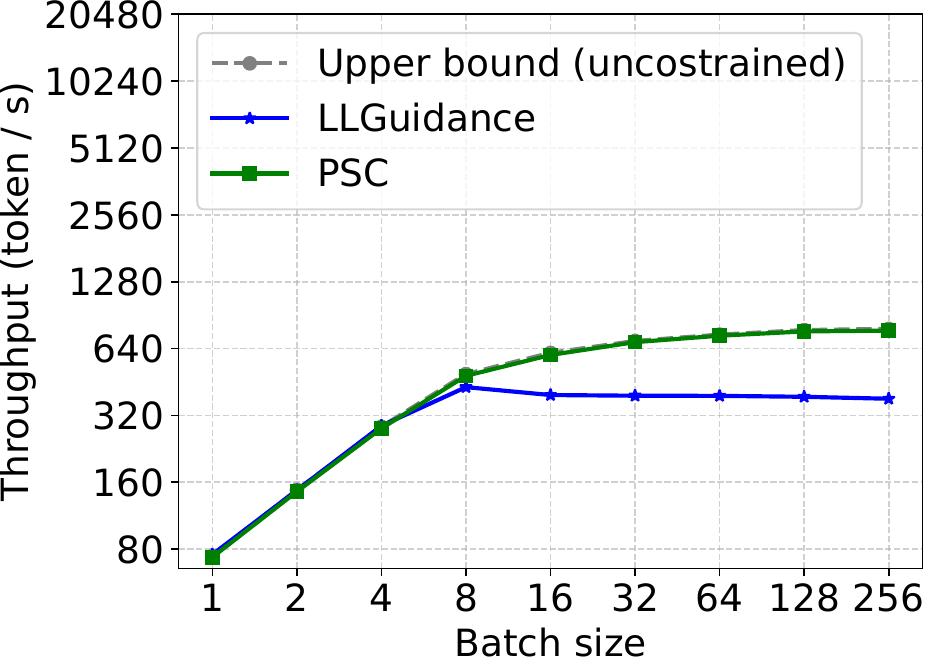}}\\

Go &
\raisebox{-.5\totalheight}{\includegraphics[width=.21\textwidth]{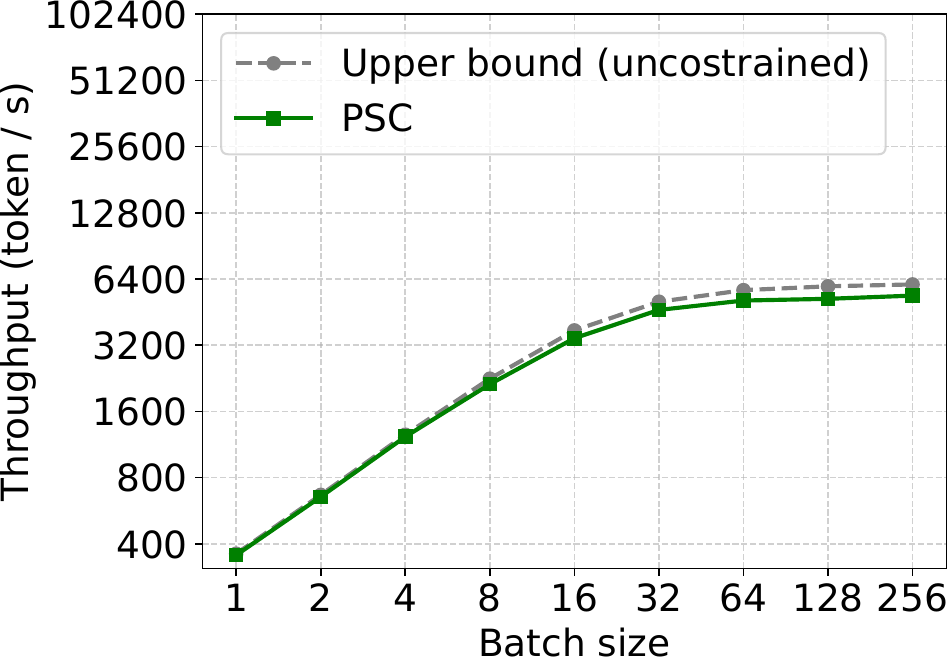}}
&
\raisebox{-.5\totalheight}{\includegraphics[width=.21\textwidth]{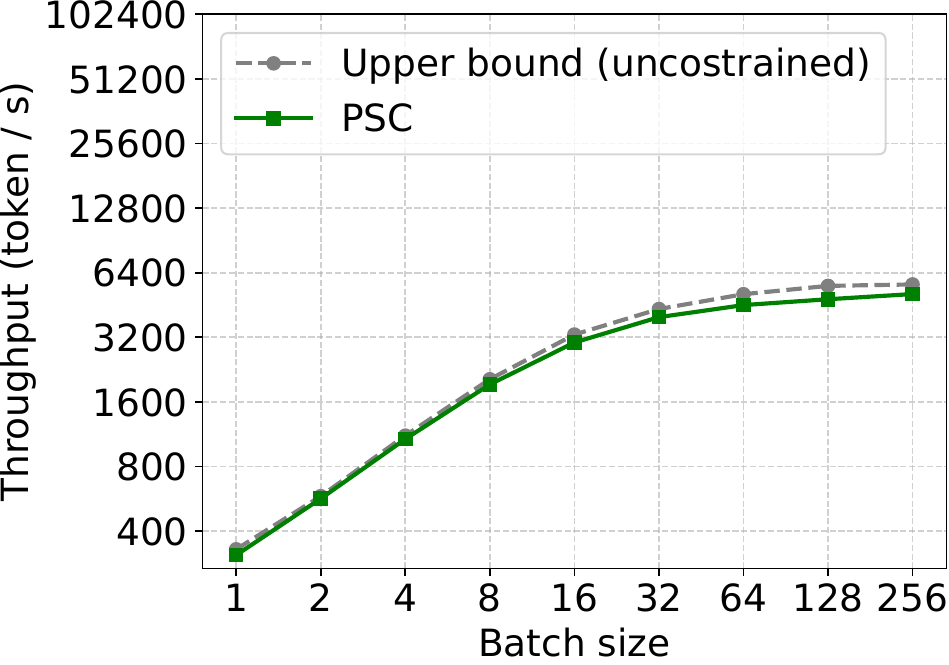}}
&
\raisebox{-.5\totalheight}{\includegraphics[width=.21\textwidth]{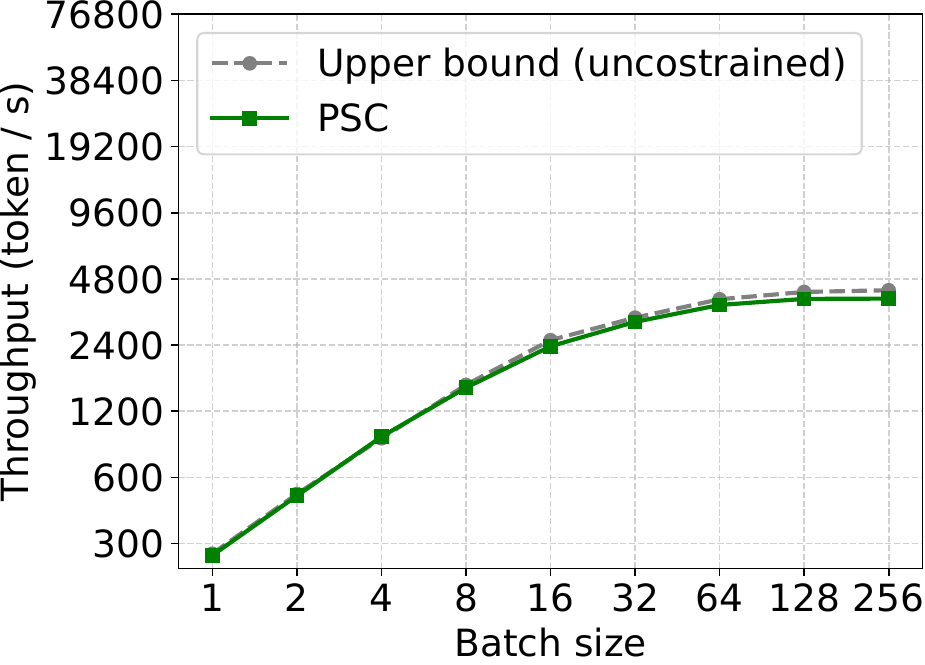}}
&
\raisebox{-.5\totalheight}{\includegraphics[width=.21\textwidth]{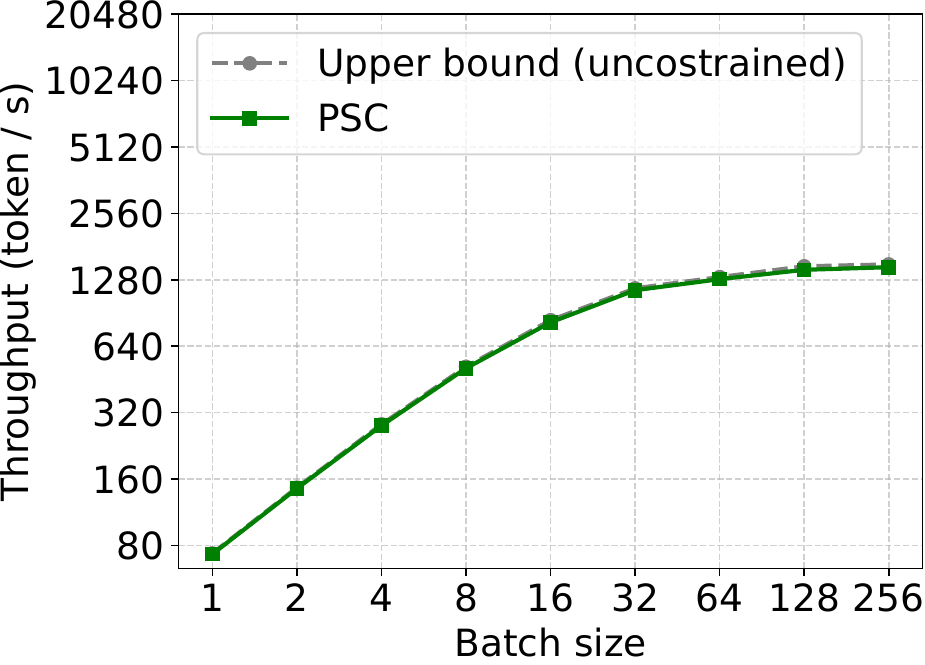}}\\

SQL &
\raisebox{-.5\totalheight}{\includegraphics[width=.21\textwidth]{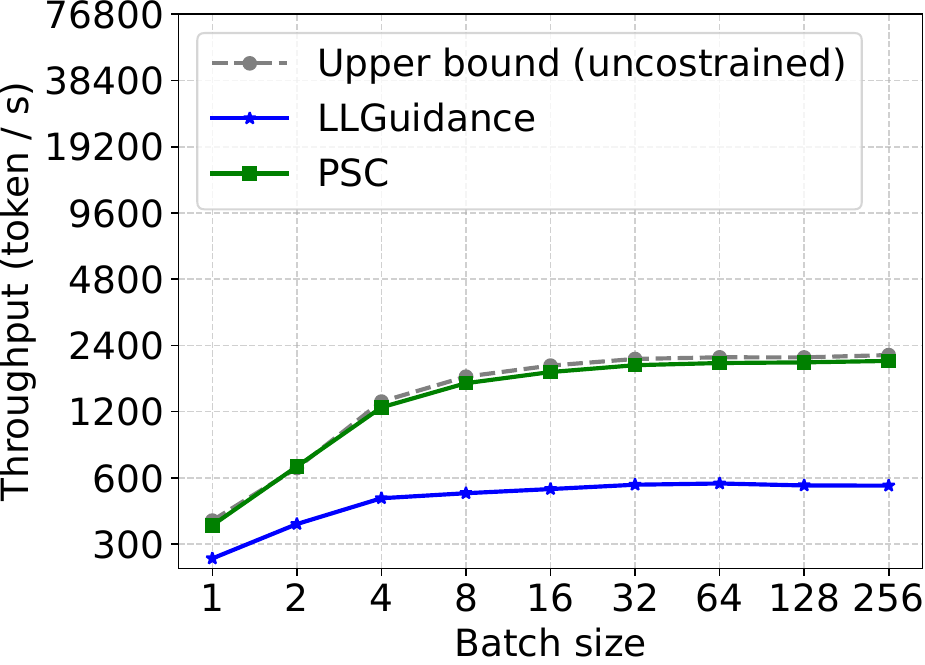}}
&
\raisebox{-.5\totalheight}{\includegraphics[width=.21\textwidth]{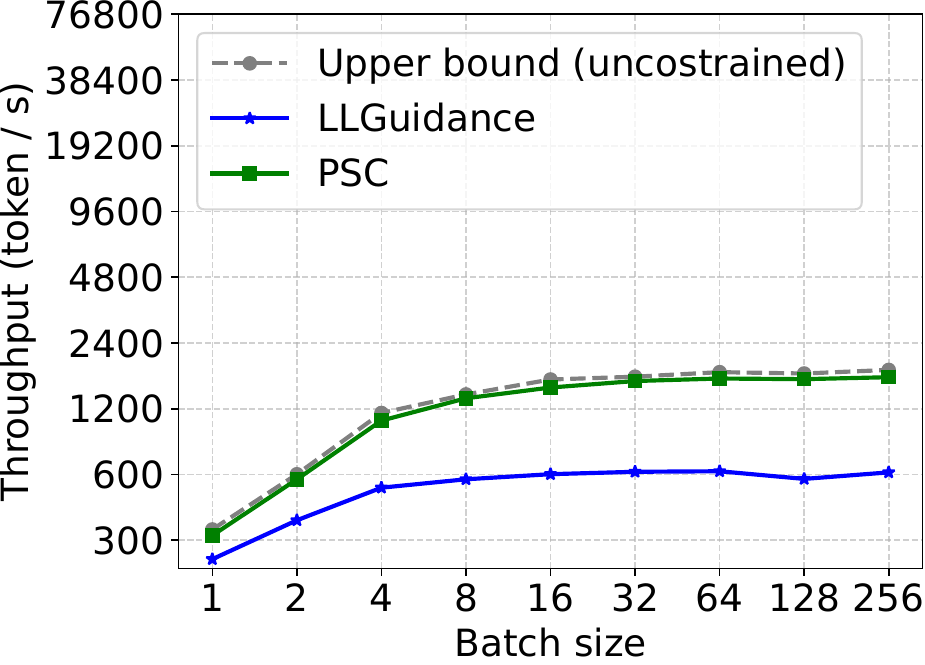}}
&
\raisebox{-.5\totalheight}{\includegraphics[width=.21\textwidth]{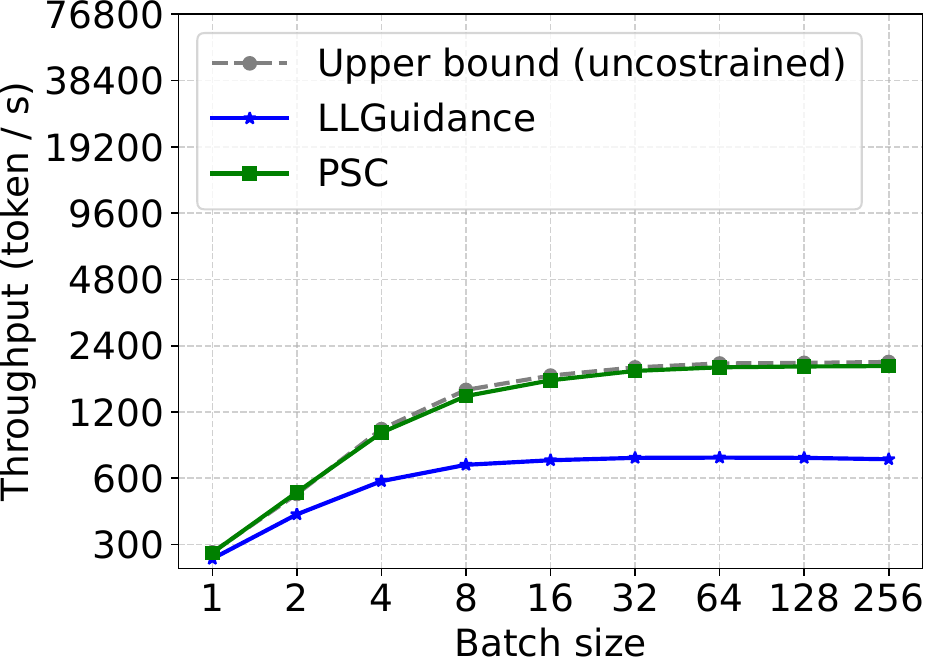}}
&
\raisebox{-.5\totalheight}{\includegraphics[width=.21\textwidth]{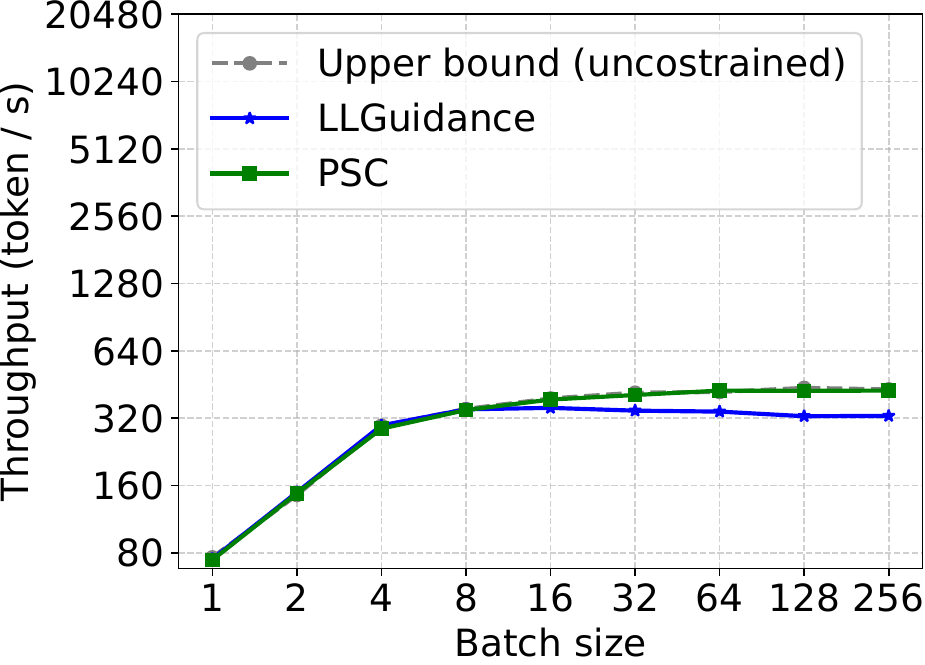}}\\

\multirowcell{1}{JSON\\Schemas} &
\raisebox{-.5\totalheight}{\includegraphics[width=.21\textwidth]{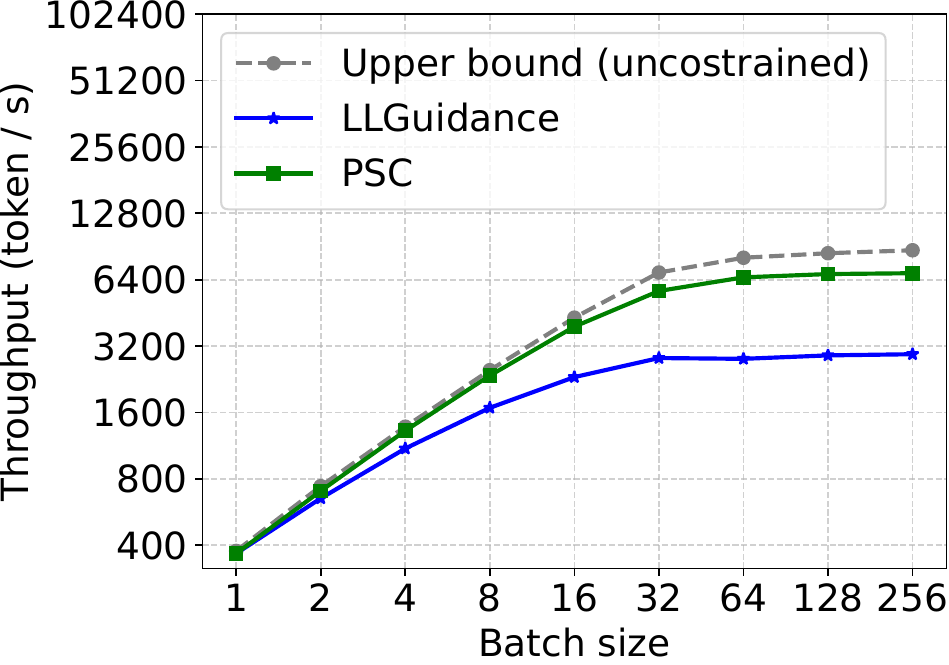}}
&
\raisebox{-.5\totalheight}{\includegraphics[width=.21\textwidth]{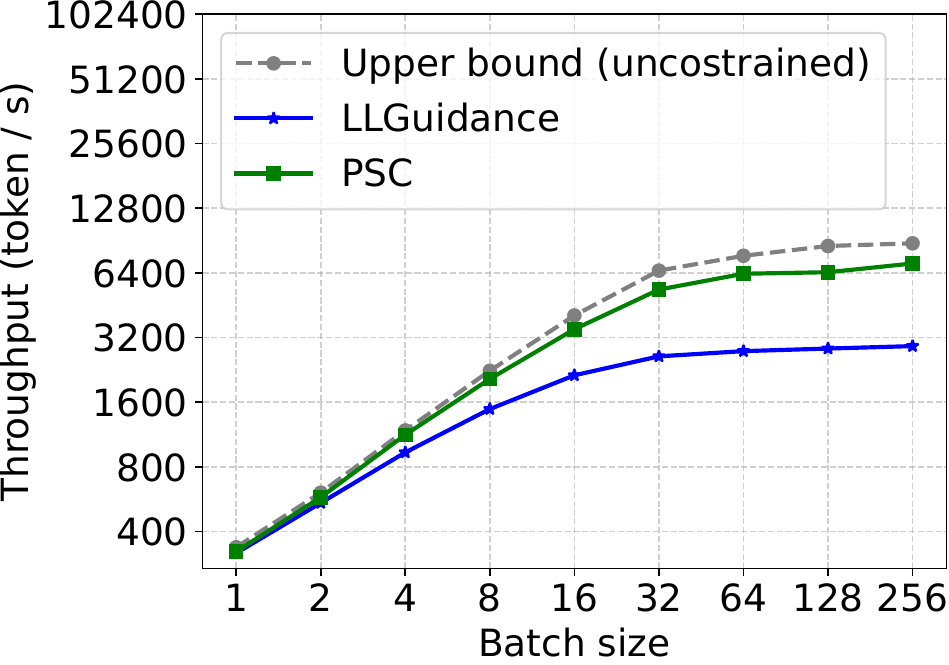}}
&
\raisebox{-.5\totalheight}{\includegraphics[width=.21\textwidth]{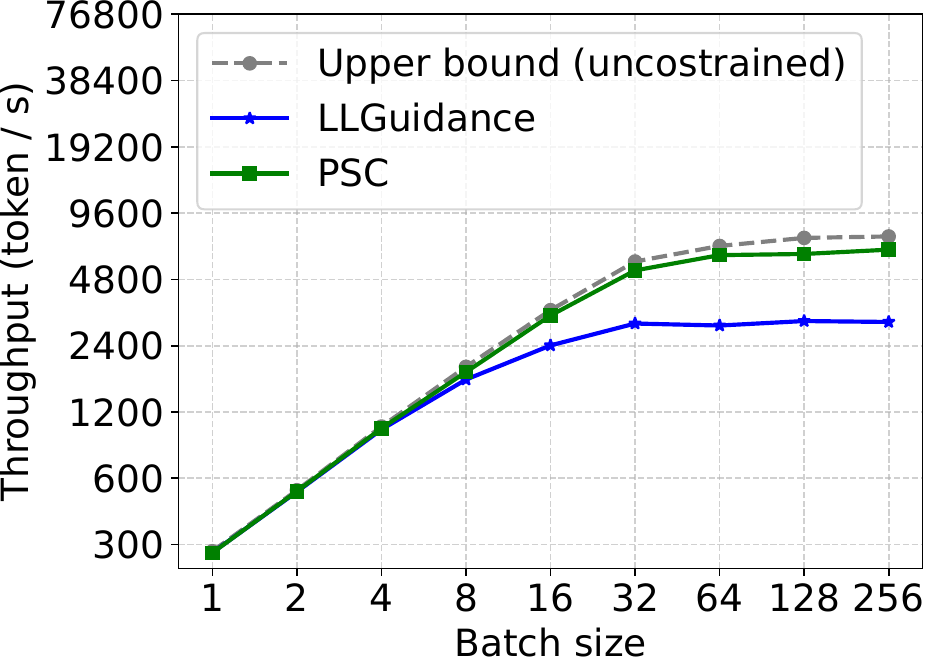}}
&
\raisebox{-.5\totalheight}{\includegraphics[width=.21\textwidth]{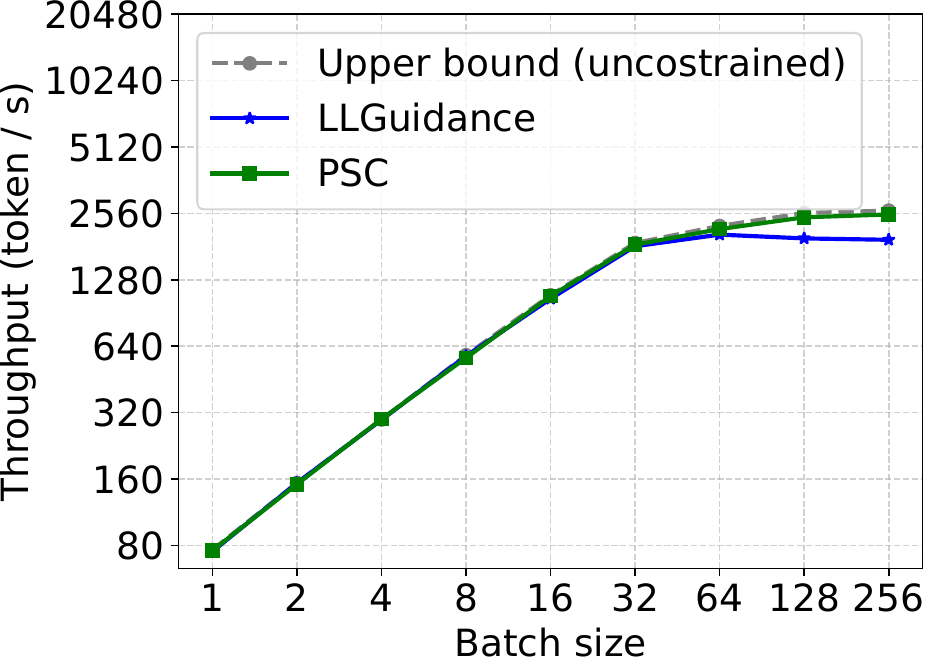}}
\end{tabular}
\caption{End-to-end throughput (tokens per second) on the different dataset using different methods on different models with various batch sizes.}
\label{fig:throughput}
\end{figure}

\setlength{\tabcolsep}{\oldtabcolsep}

%% file: chapters/experiments/downstream.tex
\subsection{\ref{rq:downstream}: Downstream Performance of Grammar-Constrained Decoding}\label{sec:downstream}

In this section, we evaluate the performance of \method{} on downstream tasks that benefit from grammar-constrained decoding.
Note that, in principle, all GCD methods calculate the same valid token masks, so their performance on downstream tasks should be similar, and we should observe similar downstream performance improvements over unconstrained decoding.
To verify this, we mainly replicate the downstream experiments from Syncode~\citep{syncode}, comparing \method{} with unconstrained decoding.

\paragraph{Settings}
We evaluate three downstream tasks: Go generation, schema-conformant JSON generation, and text-to-SQL generation.
For each task, we compare grammar-constrained decoding using \method{} (GCD) with unconstrained decoding (Standard).
We use the same models as in previous experiments: Llama~3.2~1B, Qwen2.5~0.5B and Gemma~3~270M, and use instruct-tuned versions of the models for tasks with chat-style instructions.
For most tasks, we use the pass@$k$ metric~\citep{codex} for evaluation.
When generating outputs, we use greedy decoding for pass@1, and sampling (temperature 1.0, no top-$p$ or top-$k$ filtering) for pass@$k$ where $k>1$.

\subsubsection{Go Generation}

\paragraph{Settings}
We use three datasets to evaluate the Go generation performance using GCD with \method{}:
\begin{enumerate*}
\item the Go subset of Multilingual HumanEval~\citep{mxeval,codex} dataset (160 samples);
\item the MBGP~\cite{mxeval} dataset (939 samples from the MBPP~\cite{mbpp} dataset ported from Python to Go);
\item and the Go subset of McEval~\cite{mceval} dataset (50 questions in Go).
\end{enumerate*}

\paragraph{Results}
The results are presented in Table~\ref{tab:go}.
For all the models tested, grammar-constrained decoding with \method{} using the Go grammar outperforms unconstrained decoding (Standard) on all three datasets.
This confirms that grammar-constrained decoding can effectively improve the performance of LLMs on Go generation tasks.

\input{floats/downstream/go}

\subsubsection{JSON Generation}

\paragraph{Settings}
We use the JSON-Mode-Eval~\citep{jsonmodeeval} dataset, which contains 100 samples of natural language instructions, each paired with a JSON schema and the corresponding correct JSON output.
We slightly modify the dataset for our experiments\footnote{The oracle answer of sample 39 simply copies the schema (which is valid JSON but clearly not the intended output), so we exclude this sample from evaluation.
Our script (obtained from MaskBench~\citep{maskbench}, as described in Section~\ref{sec:exp:setup}) used to generate grammars from JSON schemas fails to generate a valid grammar for certain schemas.
To address this, we replace the schemas of sample 19, 24, 27, 33, 45 and 72 with the equivalent schemas supported by our grammar generation script; the schemas of sample 1, 15, 22, 90 and 97 cannot be converted to equivalent schemas supported by our grammar generation script, so we only enforce the JSON grammar on these samples.}.
The exact schemas used in our experiments are provided in the open-sourced code.
We compare one additional method: grammar-constrained decoding using \method{} with grammars generated from the JSON schemas (GCD + Schema).
The generated JSON is considered correct if and only if it can be converted to a JSON object that exactly matches the oracle answer.
For the unconstrained decoding, we strip the code block markers, e.g. \texttt{\textasciigrave\textasciigrave\textasciigrave json} ... \texttt{\textasciigrave\textasciigrave\textasciigrave}, before checking the validity of the generated JSON.

\paragraph{Results}
The results are presented in Table~\ref{tab:json}.
For all the models tested, grammar-constrained decoding using \method{} with the grammar of the JSON schema (GCD + Schema) significantly outperforms grammar-constrained decoding with only the JSON grammar (GCD), which outperforms unconstrained decoding (Standard).
The improvement is especially significant on Llama~3.2~1B and Gemma~3~270M, with up to 10 points and 27 points absolute improvement in pass@1, respectively.
This confirms that grammar-constrained decoding can effectively improve the performance of LLMs on schema-conformant JSON generation tasks.

\input{floats/downstream/json}

\subsubsection{Text-to-SQL Generation}

\paragraph{Settings}
We use the development set of Spider~\citep{spider} dataset, containing 1,034 text-to-SQL samples, and construct the same prompts as in Syncode~\citep{syncode}.
During the experiments, we fix\footnote{We found that the SQL grammar used in Syncode~\citep{syncode} cannot parse the \texttt{NOT} operator in the boolean expressions, making some valid SQL queries unparsable.
We fixed this issue by adding the \texttt{NOT} operator to the grammar.} the SQL grammar used in Syncode~\citep{syncode}.
The fixed grammar is provided in the open-sourced code.
We use the standard execution accuracy (Exec Acc)~\citep{sql-eval} metric in this experiment, which uses greedy decoding to generate 1 SQL query for each sample.

\paragraph{Results}
The results are presented in Table~\ref{tab:json}.
For all the models tested, grammar-constrained decoding with \method{} using the SQL grammar outperforms unconstrained decoding (Standard).
This confirms that grammar-constrained decoding can effectively improve the performance of LLMs on text-to-SQL generation tasks.

\begin{tcolorbox}[size=title]
  \textbf{Answer to \ref{rq:downstream}}: Grammar-constrained decoding using \method{} effectively improves the downstream task performance of LLMs on Go generation, schema-conformant JSON generation, and text-to-SQL generation tasks.
\end{tcolorbox}

%% file: floats/downstream/go.tex
\begin{table}[t]
\centering
\caption{The pass@$k$ scores (\%) of different methods on Multilingual HumanEval Go dataset, MBGP dataset and McEval Go dataset using grammar-constrained decoding and unconstrained decoding (Standard).
For grammar-constrained decoding, we use \method{} with the Go grammar to compute the valid token masks.}
\label{tab:go}
\begin{adjustbox}{max width=\textwidth}
\begin{tabular}{c|rrrrrr|rrrrrr|rrrrrr}
\toprule
\multirow{3}{*}{Method} & \multicolumn{6}{c|}{Multilingual HumanEval Go} & \multicolumn{6}{c|}{MBGP} & \multicolumn{6}{c}{McEval Go} \\
& \multicolumn{6}{c|}{pass@$k$} & \multicolumn{6}{c|}{pass@$k$} & \multicolumn{6}{c}{pass@$k$} \\
& 1 & 3 & 5 & 10 & 20 & 50 & 1 & 3 & 5 & 10 & 20 & 50 & 1 & 3 & 5 & 10 & 20 & 50 \\\midrule
\multicolumn{13}{l}{\textcolor{gray}{$\triangleright$\textit{Llama 3.2 1B}}} \\
Standard & \textbf{5.6} & 4.0 & 5.4 & 7.4 & 9.3 & 11.9 & 13.1 & 7.3 & 10.2 & 14.9 & 20.2 & 27.8 & \textbf{6.0} & 7.8 & 11.0 & 15.4 & 18.6 & 22.0 \\
\rowcolor{\mycolor} GCD & \textbf{5.6} & \textbf{4.7} & \textbf{6.3} & \textbf{8.6} & \textbf{10.6} & \textbf{13.1} & \textbf{13.2} & \textbf{7.5} & \textbf{10.5} & \textbf{15.5} & \textbf{21.3} & \textbf{29.9} & \textbf{6.0} & \textbf{8.0} & \textbf{11.3} & \textbf{16.3} & \textbf{20.9} & \textbf{26.0} \\\midrule
\multicolumn{13}{l}{\textcolor{gray}{$\triangleright$\textit{Qwen2.5 0.5B}}} \\
Standard & \textbf{8.8} & 5.4 & 7.2 & 9.9 & 12.4 & 15.6 & 13.7 & \textbf{9.3} & 12.6 & 18.0 & 24.1 & 32.5 & \textbf{12.0} & 11.2 & 14.7 & 19.2 & 24.2 & \textbf{34.0} \\
\rowcolor{\mycolor} GCD & \textbf{8.8} & \textbf{6.0} & \textbf{7.9} & \textbf{10.5} & \textbf{13.2} & \textbf{16.3} & \textbf{14.0} & \textbf{9.3} & \textbf{12.8} & \textbf{18.3} & \textbf{24.5} & \textbf{32.6} & \textbf{12.0} & \textbf{11.6} & \textbf{15.5} & \textbf{20.6} & \textbf{25.9} & \textbf{34.0} \\\midrule
\multicolumn{13}{l}{\textcolor{gray}{$\triangleright$\textit{Gemma 3 270M}}} \\
Standard & \textbf{0.6} & 0.8 & 1.0 & 1.3 & 1.8 & 2.5 & 1.9 & \textbf{0.8} & \textbf{1.3} & 2.2 & 3.6 & 6.1 & \textbf{4.0} & 0.6 & 0.9 & 1.6 & 2.6 & 4.0 \\
\rowcolor{\mycolor} GCD & \textbf{0.6} & \textbf{1.0} & \textbf{1.3} & \textbf{1.9} & \textbf{2.9} & \textbf{5.0} & \textbf{2.0} & \textbf{0.8} & \textbf{1.3} & \textbf{2.3} & \textbf{3.7} & \textbf{6.2} & \textbf{4.0} & \textbf{1.2} & \textbf{1.7} & \textbf{2.6} & \textbf{3.6} & \textbf{6.0} \\\bottomrule
\end{tabular}
\end{adjustbox}
\end{table}

%% file: floats/downstream/json.tex
\begin{table}[t]
\centering
\caption{The results of different methods using grammar-constrained decoding and unconstrained decoding (Standard) on schema-conformant JSON generation of the JSON-Mode-Eval dataset and text-to-SQL generation of the Spider dataset.
For grammar-constrained decoding, we use \method{} to compute the valid token masks.
The GCD method uses only the JSON grammar (not specific to the schema) or the SQL grammar, while the GCD + Schema method uses the grammar generated from the JSON schema.}
\label{tab:json}
\begin{adjustbox}{max width=\textwidth}
\begin{tabular}{cc|rrrrrr|ccccc}
\toprule
\multirow{3}{*}{Model} & \multirow{3}{*}{Method} & \multicolumn{6}{c|}{JSON-Mode-Eval} & \multicolumn{5}{c}{Spider} \\
& & \multicolumn{6}{c|}{pass@$k$(\%)} & \multicolumn{5}{c}{Execution accuracy(\%)} \\
& & 1 & 3 & 5 & 10 & 20 & 50 & Easy & Medium & Hard & Extra Hard & Overall \\\midrule
\multirowcell{3}{Llama 3.2 1B} & Standard & 55.6 & 53.9 & 60.0 & 67.1 & 73.5 & 80.8 & 35.9 & 25.3 & \textbf{14.9} & \textbf{6.6} & 23.1 \\
& \chl GCD & \chl 56.6 & \chl 55.9 & \chl 61.9 & \chl 68.1 & \chl 73.7 & \chl 80.8 & \chl \textbf{38.7} & \chl \textbf{27.8} & \chl \textbf{14.9} & \chl \textbf{6.6} & \chl \textbf{24.9} \\
& \chl GCD {\scriptsize + Schema} & \chl \textbf{68.7} & \chl \textbf{70.9} & \chl \textbf{74.2} & \chl \textbf{78.0} & \chl \textbf{81.4} & \chl \textbf{84.8} & - & - & - & - & - \\\midrule
\multirowcell{3}{Qwen2.5 0.5B} & Standard & 64.7 & 66.3 & 69.7 & 73.1 & 75.9 & 80.8 & \textbf{31.0} & 24.0 & \textbf{13.2} & \textbf{6.6} & 21.1 \\
& \chl GCD & \chl 64.7 & \chl 67.0 & \chl 70.8 & \chl 74.8 & \chl 77.8 & \chl 81.8 & \chl \textbf{31.0} & \chl \textbf{24.2} & \chl \textbf{13.2} & \chl \textbf{6.6} & \chl \textbf{21.2} \\
& \chl GCD {\scriptsize + Schema} & \chl \textbf{67.7} & \chl \textbf{69.8} & \chl \textbf{73.0} & \chl \textbf{77.2} & \chl \textbf{81.2} & \chl \textbf{85.9} & - & - & - & - & - \\\midrule
\multirowcell{3}{Gemma 3 270M} & Standard & 28.3 & 31.3 & 33.5 & 36.1 & 38.9 & 44.4 & 0.4 & 0.4 & 0.0 & 0.6 & 0.4 \\
& \chl GCD & \chl 29.3 & \chl 33.2 & \chl 35.6 & \chl 38.6 & \chl 41.4 & \chl 45.5 & \chl \textbf{1.6} & \chl \textbf{1.3} & \chl \textbf{0.6} & \chl \textbf{1.8} & \chl \textbf{1.4} \\
& \chl GCD {\scriptsize + Schema} & \chl \textbf{56.6} & \chl \textbf{62.1} & \chl \textbf{64.2} & \chl \textbf{66.2} & \chl \textbf{67.6} & \chl \textbf{68.7} & - & - & - & - & - \\\bottomrule
\end{tabular}
\end{adjustbox}
\end{table}

%% file: chapters/discuss.tex
\section{Discussion}\label{sec:discuss}

In this section, we mainly discuss the preprocessing overhead of \method{}, and analyze the applicability of \method{} in practical scenarios.
We split the users of \method{} into two types: \textbf{preprocessing providers} who perform the preprocessing, and \textbf{users} who use the preprocessing results to perform grammar-constrained decoding.
The preprocessing providers can be the same as or different from the users. For example, the preprocessing providers can be the model developers who preprocess for common grammars and share the results with users, or they can be the users themselves who preprocess for their specific grammars.

\input{floats/break-even-demo}
\subsection{Overhead for Preprocessing Providers}\label{sec:preprocess}

The preprocessing overhead of \method{} includes the time and memory footprint of preprocessing.
The preprocessing here includes all the steps described in Algorithm~\ref{alg:offline}, as well as the precomputation of all the possible token masks for the DFA states.
The raw results of preprocessing overhead are presented in Table~\ref{tab:preprocess}.

\input{floats/preprocess}

\subsubsection{Raw Results of Preprocessing Overhead}

\paragraph{Results on JSON schemas} The average preprocessing time is around half to one minute per schema, and the memory footprint is around 3 GiB for Llama~3 and Qwen2.5, and around 6 GiB for Gemma~3.
This is quite practically reasonable, allowing for quick adaptation to new schemas.

\paragraph{Results on programming language} The preprocessing time ranges from around 8 minutes for Java to around 1.3 hours for SQL.
The memory footprint ranges from around 40 GiB for Java to around 250 GiB for SQL.

While the preprocessing overhead for programming language grammars is higher than that for JSON schemas, it is still acceptable as it is only needed \textbf{once} per grammar and vocabulary pair.
The grammars of programming languages are typically stable, and the vocabulary is commonly shared by the same series of models, so the preprocessing only needs to be done infrequently by the preprocessing providers who have the necessary resources and expertise.
As a demonstration, we provide all the preprocessing results used in our experiments in our replication package.

\subsubsection{Applicability Analysis for Decoding Users to Preprocess}

\paragraph{Applicability analysis: preprocessing time}
Sometimes the users of \method{} may also need to perform the preprocessing themselves, for example when they want to use \method{} for self-defined JSON schema.
In this case, they should consider how much time they can save  during decoding by using \method{} compared to using other GCD methods, and compare it with the preprocessing time to decide whether to use \method{}.
We provide a break-even point analysis to help users make this decision,
i.e., the total time users need to use \method{} for decoding to make it more time-efficient than using LLGuidance, our fastest baseline in the experiments.

The results are presented in Table~\ref{tab:preprocess}.
For JSON schemas, the preprocessing time is generally small, and the break-even point is roughly half minutes of decoding time.
Therefore, doing preprocessing with \method{} by users for new JSON schemas is generally recommended.
For programmaing languages, though, the perprocessing time is higher (several minutes to more than an hour), so users may consider using the preprocessing results provided by the preprocessing providers, e.g. model developers or other users, to avoid the preprocessing overhead.

\paragraph{Applicability analysis: memory usage}
The memory usage during preprocessing is generally higher than that during decoding.
But preprocessing can be performed on cloud instances with large memory capacity if needed.
For example, AWS offers instances with 500 GiB of memory for on-demand usage at around 2.5 USD per hour.
Also, the memory usage can potentially be reduced by optimizing the implementation, such as using more efficient data structures or algorithms.

\subsection{Overhead for Decoding Users}\label{sec:runtime_overhead}

\input{floats/runtime_overhead}

\method{} has a runtime overhead for maintaining the parser stack and the token masks, as described in Section~\ref{sec:method}.
The user also saves the preprocessing results on disk, which is also part of the overhead for users.\footnote{The preprocessing results are compressed and loaded as needed using zstandard \citep{zstd}.}
Table~\ref{tab:runtime_overhead} summarizes the runtime overhead of \method{} for different grammars and models.
Overall, the disk usage is generally small, and the runtime memory overhead is generally managable for practical applications.
The runtime time overhead, as shown in Section~\ref{sec:exp:overhead}, is generally negligible, making \method{} a practical choice for grammar-constrained decoding in real-world applications.

\begin{tcolorbox}[size=title]
  \textbf{Takeaway}:
  \method{} preprocessing is affordable for preprocessing providers. Decoding users should consider the break-even point analysis when performing the preprocessing themselves.
\end{tcolorbox}

\subsection{Threats to Validity}
\label{sec:threats}

In this section, we discuss potential threats to the validity of our experimental results.
\begin{enumerate}[wide]
  \item \textbf{Generality of grammars.} \method{} can only handle deterministic context-free grammars (DCFGs), and our theoretical analysis is based on DCFGs,
    while some existing methods can handle more general grammar formats, such as non-deterministic context-free grammars (CFGs).
    In our experiments, we only consider DCFGs for fair comparison.
    However, many practical grammars, such as most programming languages and JSON schemas, can be expressed as DCFGs.
  \item \textbf{Generality of tokenization.} Different models may have different vocabulary that affect the performance of \method{} and baselines.
    We mitigate this threat by using multiple LLM series with different vocabulary sizes and tokenization schemes.
    In addition, the online time complexity of \method{} is independent of the vocabulary size, so it only affects the preprocessing overhead.
  \item \textbf{Generality of model sizes.} Different-sized models may have different decoding throughputs that affect the relative performance of \method{} and baselines.
    Generally, larger models have lower decoding throughput due to their increased computational requirements, making the throughput differences less significant.
    We try to mitigate this threat by evaluating the throughputs on Qwen2.5~7B in Section~\ref{sec:exp:throughput} in addition to the smaller models used in other experiments.
  \item \textbf{LLM randomness.} The evaluation results in RQ4 might be effected by the randomness of LLMs.
  We mitigate this threat by using deterministic greedy decoding in pass@1, and using 50 samples per question to calculate pass@$k$ for $k>1$ to reduce the variance.
  When comparing \method{} with unconstrained decoding, we use the same decoding settings for both methods with the same default prompts, so the relative improvement should not be significantly affected by randomness.
  \item \textbf{Data leakage.} The evaluation results in RQ4 might be effected by data leakage in the training data of the LLMs.
  We try to mitigate the threat of data leakage by only comparing \method{} with unconstrained decoding on the same model with the same decoding settings, so the relative improvement should not be significantly affected by data leakage.
\end{enumerate}

%% file: floats/break-even-demo.tex
\begin{wrapfigure}[20]{r}{0.5\textwidth}
\input{floats/break-even-demo-wrapped}
\end{wrapfigure}

%% file: floats/break-even-demo-wrapped.tex
\centering
\begin{tikzpicture}

\draw[->] (0,0) -- (4.5,0) node[right] {Total time};
\draw[->] (0,0) -- (0,4) node[right, align=left] {Total tokens\\generated};

\draw[gray, dashed] (0,0) -- (2.193,4) node[pos=0.9, right, align=left] {Unconstrained\\7524.7 tokens/s};
\draw[blue, line width=1pt] (0,0) -- (4.5,3.357) node[below right, align=left] {LLGuidance\\3077.0 tokens/s};
\draw[black!50!green, line width=1pt] (2.123,0) --(4.5,3.777) node[right, align=left] {\method{}\\6553.1 tokens/s};

\draw[red, dashed] (4.00,2.985) -- (4.00,0);
\draw[red, dashed] (4.00,2.985) -- (0,2.985);

\draw[black!50!green, line width=2pt] (0,0) -- node[below, align=center] {preprocessing\\28.3 s} (2.123,0);

\draw[red, line width=2pt] (2.123,0) -- node[below, align=center] {break-even\\25.1 s} (4.00,0);
\draw[red, line width=2pt] (0,2.985) -- node[right, align=left] {164.2k\\tokens} (0,0);

\end{tikzpicture}
\caption{The break-even point of \method{} versus LLGuidance, based on data from Llama~3.2~1B on JSON Schemas. Different colors represent different methods. The break-even point is where the two lines intersect, indicating the total time users need to use \method{} to amortize the preprocessing time compared to using LLGuidance \textbf{for doing the preprocessing by themselves}.}
\label{fig:break_even}

%% file: floats/preprocess.tex
\begin{table}[t]
\centering
\caption{\textbf{Preprocessing} overhead and break-even point for preprocessing against LLGuidance of \method{} on different grammars. The break-even point for preprocessing indicates the total time (in seconds) the user needs to use \method{} to amortize the preprocessing time compared to using LLGuidance.}
\label{tab:preprocess}
\begin{tabular}{cl|cccc}
\toprule
\multirow{2}{*}{Metrics} & \multirow{2}{*}{Model} & \multicolumn{4}{c}{Grammar} \\
& & Java & Go & SQL & JSON Schemas \\\midrule
\multirowcell{3}{Time Used in\\Preprocessing (seconds)} & Llama 3 & 466.9 &	1171.7 & 4662.7 & 28.3 \\
& Qwen2.5 & 464.5 & 1166.8 & 4770.0 & 28.6 \\
& Gemma 3 & 472.1 & 1362.2 & 1367.4 & 53.2 \\\midrule
\multirowcell{3}{Memory Usage during\\Preprocessing (GiB)} & Llama 3 & 40.8 &	87.7 & 255.3 & 3.04 \\
& Qwen2.5 & 40.4 & 86.6 & 254.2 & 3.13 \\
& Gemma 3 & 36.0 & 88.8 & 188.7 & 5.95 \\\midrule\midrule
\multirowcell{3}{Break-even Point\\(seconds)} & Llama 3 1B & 146.6 & - & 2826.2 & 25.1 \\
& Qwen2.5 0.5B & 101.5 & - & 2760.7 & 20.1 \\
& Gemma 3 270M & 67.2 & - & 508.5 & 40.0 \\\bottomrule
\end{tabular}
\end{table}

%% file: floats/runtime_overhead.tex
\begin{table}[bt]
\centering
\caption{\textbf{Runtime} overhead of \method{} on different grammars.}
\label{tab:runtime_overhead}
\begin{tabular}{cl|cccc}
\toprule
\multirow{2}{*}{Metrics} & \multirow{2}{*}{Model} & \multicolumn{4}{c}{Grammar} \\
& & Java & Go & SQL & JSON Schemas \\\midrule
\multirowcell{3}{Runtime memory for\\DFA $\mathcal A$} & Llama 3 & 162.8  MiB & 752.2  MiB & 3.42  GiB & 296.8  KiB \\
& Qwen2.5 & 165.2  MiB & 763.4  MiB & 3.37  GiB & 284.6  KiB \\
& Gemma 3 & 124.3  MiB & 932.9  MiB & 1020.4  MiB & 288.7  KiB \\\midrule
\multirowcell{3}{Runtime memory for\\precomputed masks} & Llama 3 & 227.3  MiB & 1.68  GiB & 2.82 GiB & 1.84 MiB \\
& Qwen2.5 & 264.5 MiB & 1.98 GiB & 3.33 GiB & 2.02 MiB \\
& Gemma 3 & 457.4 MiB & 6.32 GiB & 2.78 GiB & 3.31 MiB \\\midrule\midrule
\multirowcell{3}{Disk Space of\\Preprocessing Results} & Llama 3 & 13.27 MiB & 28.16 MiB & 58.95 MiB & 0.54 MiB \\
& Qwen2.5 & 12.70 MiB & 28.37 MiB & 57.60 MiB & 0.51 MiB \\
& Gemma 3 & 13.13 MiB & 33.52 MiB & 24.04 MiB & 0.76 MiB \\\bottomrule
\end{tabular}
\end{table}

%% file: chapters/conclude.tex
\section{Conclusion}
\label{sec:conclude}

In this paper, we present \method{}, a novel approach for grammar-constrained decoding.
By constructing the exact requirements on the parser stack for each vocabulary token, \method{} can determine the valid tokens at each decoding step by a single pass through the parser stack.
Our experimental results demonstrate that \method{} achieves significant speedup over existing methods, and the end-to-end throughput of \method{} approaches that of unconstrained decoding.
This makes \method{} a practical choice for real-world applications that require grammar-constrained decoding.

One limitation of \method{} is that it requires preprocessing to construct the DFA for the given grammar and vocabulary.
We provide a break-even point analysis to help users decide whether to do preprocessing by themselves, which depends on the expected decoding usage.
In future, we plan to explore other types of grammars and constraints that can be efficiently handled by \method{}, and other optimizations to further improve its efficiency and scalability.

%% file: chapters/reproduce.tex
\section*{Data Availability}

We open-source the code to facilitate reproducibility of our results at \url{https://github.com/Gompyn/PSC}.
It includes the implementation of \method{}, the datasets used in the experiments, the scripts to run the experiments and generate the results in the paper, and the preprocessing results built from the grammars used in our experiments.